\RequirePackage{fix-cm}
\documentclass[smallextended]{svjour3}       
\smartqed  
\usepackage{graphicx}
\usepackage{url}
\usepackage{amsmath}
\usepackage{amssymb}
\usepackage{amsfonts}
\usepackage{tikz}
\usepackage{color}
\usetikzlibrary{babel}
\usetikzlibrary{decorations.pathreplacing}
\newcommand{\Z}{\mathbb{Z}}

\newcommand{\R}{\mathbb{R}}
\newcommand{\C}{\mathbb{C}}

\newcommand{\del}{\partial}

\newcommand{\XPlusHat}{X_{\mathrm{cyl}}}

\tikzset{
  XPlus/.pic={
    \draw (0,0.5) -- (4.5,0.5) arc [start angle=150, end angle=-150, radius=1] -- (0,-0.5);
    \node at (5.366,0) {$X_+$};
  },
  XMinus/.pic={
    \draw (0,0.5) -- (-4.5,0.5) arc [start angle=30, end angle=330, radius=1] -- (0,-0.5);
    \node at (-5.366,0) {$X_-$};
  },
  Y/.pic={
    \draw [ultra thick] (0,-0.5) -- (0,0.5);
    \node at (0,0) [right]{$Y$};
  },
  finiteNeckY/.pic={
    \filldraw [very nearly transparent] (-2,0.5) -- (2,0.5) -- (2,-0.5) -- (-2,-0.5);
    \draw [ultra thick] (0,-0.5) -- (0,0.5);
    \draw
      [->] (-2.2,0) -- (2.2,0) node [right]{$u$};
    \draw
      (-2,0.06) -- (-2,-0.06) node [below]{$-4$}
      (2,0.06) -- (2,-0.06) node [below]{$4$};
  },
  infiniteHalfCylinderY/.pic={
    \draw [ultra thick] (0,-0.5) -- (0,0.5);
    \draw (-7,0.5) -- (0,0.5) -- (0,-0.5) -- (-7,-0.5);
    \filldraw [very nearly transparent] (-7,0.5) -- (0,0.5) -- (0,-0.5) -- (-7,-0.5);
    \node at (-5.366,0) {$(-\infty,0) \times Y$};
  },
  verticalX/.pic={
    \draw
      (0,1) arc [start angle=90, end angle=-60, radius=0.25] -- (0.125,0)
      (0,-1) arc [start angle=-90, end angle=60, radius=0.25] -- (0.125,0);
    \draw [dotted]
      (0,1) arc [start angle=90, end angle=240, radius=0.25] -- (-0.125,0)
      (0,-1) arc [start angle=270, end angle=120, radius=0.25] -- (-0.125,0);
  },
  verticalXPlusHat/.pic={
    \draw (0,1) arc [start angle=90, end angle=-60, radius=0.25] -- (0.125,-2);
    \draw [dotted] (0,1) arc [start angle=90, end angle=240, radius=0.25] -- (-0.125,-2);
  },
  infiniteCylinderX/.pic={
    \draw [gray, ->] (-6.5,0) -- (6.5,0) node [right]{$\R$};
    \draw [gray, ->] (0,-1.5) -- (0,1.5) node [above]{$Y$};
    \draw
      (-6,1) -- (6,1)
      (-6,-1) -- (6,-1)
      pic at (-4.5,0) {verticalX}
      pic at (4.5,0) {verticalX};
  },
  infiniteCylinderXPlusHat/.pic={
    \draw [->, gray] (-6.5,0) -- (6.5,0) node [right]{$\R$};
    \draw [->, gray] (0,-2) -- (0,1.5) node [above]{$\XPlusHat$};
    \draw
      (-6,1) -- (6,1)
      pic at (-4.5,0) {verticalXPlusHat}
      pic at (4.5,0) {verticalXPlusHat};
  },
}

\begin{document}

\begin{flushright}
  OU-HET-1073
  \end{flushright}

\title{Mod-two APS index and domain-wall fermion\thanks{} 
}



\author{Hidenori Fukaya         \and
  Mikio Furuta \and Yoshiyuki Matsuki \and Shinichiroh Matsuo \and Tetsuya Onogi\and Satoshi Yamaguchi\and Mayuko Yamashita  
}


\institute{H. Fukaya, Y. Matsuki, T. Onogi, S. Yamaguchi\at
              Department of Physics, Osaka University, Osaka, Japan \\
              \email{hfukaya@het.phys.sci.osaka-u.ac.jp}\\
              \url{http://www-het.phys.sci.osaka-u.ac.jp/~hfukaya/}\\
              \email{ymatsuki@het.phys.sci.osaka-u.ac.jp}\\
\email{onogi@phys.sci.osaka-u.ac.jp}\\
  \email{yamaguch@het.phys.sci.osaka-u.ac.jp}
           \and
           M. Furuta \at
           Graduate School of Mathematical Sciences, The University of Tokyo, Tokyo, Japan \\
           \email{furuta@ms.u-tokyo.ac.jp}
           \and
           S. Matsuo \at
           Graduate School of Mathematics, Nagoya University, Nagoya, Japan\\
           \email{shinichiroh@math.nagoya-u.ac.jp}\\
           \url{https://www.math.nagoya-u.ac.jp/~shinichiroh/}
           \and
           M. Yamashita \at
           Research Institute for Mathematical Sciences, Kyoto University, Kyoto, Japan\\
           \email{mayuko@kurims.kyoto-u.ac.jp}
}

\date{}

\maketitle

\begin{abstract}
  We reformulate the mod-two Atiyah-Patodi-Singer (APS) index
  in a physicist-friendly way using the domain-wall fermion.
  Our new formulation is given on a closed manifold, which is
  extended from the original manifold with boundary,
  where we instead give a fermion mass term
  changing its sign at the location of the original boundary.
  This new setup does not need the APS boundary condition, which is non-local.
  A mathematical proof of equivalence between the two different formulations
  is given by two different evaluations of the same
  index of a Dirac operator on a higher dimensional manifold.
  The domain-wall fermion allows us to separate the edge and bulk
  mode contributions in a natural but not in a gauge invariant way,
  which offers a straightforward description of the global anomaly inflow.
\keywords{Index theorem \and domain-wall fermion \and anomaly inflow}
\end{abstract}

\section{Introduction}
\label{intro}

Anomaly \cite{Adler:1969gk,Bell:1969ts} 
has played an important role in studying the low-energy dynamics of
gauge theories, since it is always caused by (nearly) massless fields
that describes the infra-red physics.
As the anomaly is related to topology and thus invariant under
the renormalization group flow, we can obtain non-trivial
consequences which cannot be analyzed by perturbation.
For example, 't Hooft \cite{tHooft:1979rat} 
showed that the breaking pattern of
the chiral symmetry in QCD is quite limited. 
Also, from anomaly in electro-weak interaction of quarks,
we can determine the coefficient of
the Wess-Zumino-Witten term \cite{Wess:1971yu,Witten:1983tw} 
in pion effective Lagrangian, which agrees well with experiments.

If a theory has anomaly in its gauge invariance,
the theory is considered to be inconsistent, 
and cannot describe physics.
However, the inconsistency may be cured by
extending the theory to higher dimensions.
For example, the anomalous four-dimensional chiral fermion
can be embedded to five-dimensional vector-like gauge theory.
In such cases, the anomaly is identified as
the gauge current absorbed into
the extra-dimensions \cite{Callan:1984sa}. 
This is called the anomaly inflow \cite{Witten:2015aba},
which is recently widely studied not only in particle physics
\cite{Kurkov:2018pjw,Tachikawa:2018njr,Vassilevich:2018aqu,Garcia-Etxebarria:2018ajm,Yonekura:2019vyz,Hsieh:2020jpj,Hamada:2020mug}
but also in condensed matter physics
\cite{Gromov:2015fda,Metlitski:2015yqa,Seiberg:2016rsg,Tachikawa:2016xvs,Freed:2016rqq,Yu:2017uqt,Hasebe:2016tjg,Yonekura:2018ufj,Yao:2019ggu}.


Let us call the massless fermion on the original
(even-dimensional) manifold the edge mode,
and that lives in extra-dimension the bulk mode, which is massive or gapped.
The anomaly inflow caused by the edge mode
is cancelled by the bulk mode.
This phenomenon matches well with the so-called
bulk-edge correspondence \cite{Hatsugai} 
of topological insulators.
When the bulk fermion
has anomaly on the boundary, the edge mode
having the same anomaly with opposite sign must appear.
This realization of the bulk-edge correspondence
is valid for interacting fermions.

In \cite{Witten:2015aba}, 
the notion of the global anomaly is
extended using the anomaly inflow.
The traditional argument on the global anomaly \cite{Witten:1982fp} is
given by one-parameter family of gauge fields
which connects two gauge equivalent configurations
in $d$ dimensions. One can treat this one-parameter
as an extra dimension and when the extended $d+1$-dimensional
theory has a non-trivial topology,
the phase of the chiral fermion determinant cannot be
uniquely determined.
From the anomaly inflow point of view,
this standard set-up is limited to a torus called a mapping torus.
In \cite{Witten:2015aba,Witten:2016cio,Wang:2018qoy}, it was claimed that the global anomaly
should be extended to the case of any $d+1$-dimensional manifold.
If the phase of the fermion determinant depends on the
structure of the bulk manifold, we should regard the
theory anomalous in that the phase 
cannot be uniquely determined with $d$-dimensional information alone.

More concretely, the anomaly inflow is generalized 
by the $\eta$ invariant of a Dirac operator on the $d+1$-dimensional manifold \cite{Dai:1994kq,Freed:2014iua},
where a nontrivial boundary condition known as the APS boundary
condition \cite{APS} is imposed. (See \cite{Yonekura:2016wuc} for a physical description.)
However, the appearance of the APS boundary condition
is somewhat a puzzle in physics as it is non-locally imposed,
and therefore, it is unlikely to be realized in any physical fermion systems.
Moreover, the APS boundary condition allows no edge-localized mode
to exist, which makes it difficult to
separate the $\eta$-invariant into bulk and edge contributions.

Yonekura and Witten \cite{Witten:2019bou}
gave a physical reason why the
APS boundary condition should be introduced
rotating the normal direction to the surface
to the ``time'' direction
and regarding the APS condition as an intermediate state.
The unphysical properties cancel out between bra and ket states.
But their set-up loses spacial boundary and
it is still difficult to understand the role of the edge or bulk
modes separately in the total $\eta$ invariant.

We have been investigating more physicist-friendly alternative
understanding of the anomaly inflow without introducing
the unphysical boundary conditions.
For pseudo-real fermions, the $\eta$ invariant
is reduced to an integer called the APS index \cite{APS}.
In Refs. \cite{Fukaya:2017tsq,Fukaya:2019qlf,Fukaya:2019myi}
we succeeded in reformulating this APS index
using the domain-wall fermions \cite{Callan:1984sa,Jackiw:1975fn,Kaplan:1992bt}.
We add an ``outside'' to the original boundary and consider
a closed manifold in which the two domains are separated
by a wall. Interestingly, only the half region is shared by the 
manifold on which the original APS index is formulated
but the same index is obtained.
Although the location of the domain-wall coincides
with the boundary for the original APS, the boundary conditions
imposed on fermions are totally different.

For the mod-two version of the APS index, the issue is
more difficult, because the index can not be expressed by
any integral of local curvature functions, and  
no natural way is known to
separate the edge and bulk contributions.
As already mentioned, the APS boundary condition
allows no edge-localized mode to exist.

In this work, we extend our formulation in \cite{Fukaya:2017tsq,Fukaya:2019qlf}
to the mod-two type index which describes the sign of
the real fermion determinant.
We will show below a mathematical relation between
the domain-wall fermion determinant defined a closed manifold
to the APS index of the massless Dirac operator
given on the half of the manifold with boundary,
whose location coincides the domain-wall.
In contrast to our previous work limited to even dimensional bulk,
this work can apply to any dimensions.

The rest of the paper is organized as follows.
In Sec.\ref{sec:global}, we review the global anomaly
originally found in  \cite{Witten:1982fp}, as well as
recent development leading to the mod-two APS index,
which however, requires an unphysical boundary condition.
In Sec.\ref{eq:domainwall}, we summarize our previous
work where we achieved an alternative expression of the
standard APS index using domain-wall fermions 
without introducing any non-local conditions.
Then we mathematically prove that the mod-two APS index
can also be expressed by the domain-wall fermion
Dirac operator in Sec.\ref{sec:main},
and describe how the bulk-edge correspondence of
the anomaly is embedded in the index in Sec.\ref{sec:inflow}.
In Sec.\ref{sec:summary}, we give a summary and
discuss possible applications to higher-order topological
insulators and lattice gauge theory.

\section{Review of global anomaly}
\label{sec:global}
In this section, we review the global anomaly,
where the mod-two index theorem plays a key role.
Starting from the Witten's $SU(2)$ anomaly \cite{Witten:1982fp},
we also discuss a modern view of the anomaly
as the current inflow to the higher dimensional bulk.
In this point of view, the anomaly can be identified
as the $\eta$ invariant of the massless Dirac operator
on a manifold with boundary, as it was shown in \cite{Dai:1994kq,Freed:2014iua}
that the $\eta$ invariant satisfies properties
required to describe the topological field theory
on the manifold, which appears as an effective action
of the bulk fermions.

The mod-two APS index naturally appears as a
special case of the $\eta$ invariant.
However, it requires a non-local boundary condition
on the fermion fields, which cannot be
directly applied to the physical fermion systems.

\subsection{Global anomaly and mod-two index}

In \cite{Witten:1982fp}, the first example
of global gauge anomaly was shown, in which the sign of a
Weyl fermion path-integral in the fundamental representation
of the $SU(2)$ gauge group cannot be determined
in a gauge invariant way.
The same discussion applies to general Weyl or Majorana fermions
whose Dirac operator is real and anti-symmetric.

Let us consider a real Dirac operator $D_X$ on a manifold $X$
and assume that it has no zero eigenvalue.
The complex conjugate\footnote{In this work, we denote the complex conjugate by the superscript $*$
  and the Hermitian conjugate by $\dagger$.} of $D_X$ is given as $D_X^*=C D_X C^{-1}$
with a unitary symmetric operator $C$~\footnote{
$C$ may contain a non-trivial operator on the gauge fields.
  }.
Every non-zero eigenvalue of it makes a $\pm$ pair
since for $D_X\phi = i\lambda \phi$, we have $D_XC^{-1}\phi^* =-i\lambda C^{-1}\phi^*$
(where $\lambda$ is real).
The Weyl or Majorana fermion Lagrangian is expressed as
\begin{eqnarray}
\mathcal{L}=\frac{1}{2}\psi^T CD_X \psi,
\end{eqnarray}
with a Grassmannian variable $\psi$.
One can choose a basis so that $C=1$
and $D_X$ is real anti-symmetric operator.
In this basis, the path-integral is 
the Pfaffian of $D_X$, or ${\rm Pf} D_X$, which ends up with
a product of half of eigenvalues taking one from all eigenvalue pairs.
Since $\det D_X = ({\rm Pf} D_X)^2$ is real and positive,
${\rm Pf} D_X$ is real. This means that there is no perturbative
gauge anomaly, always  appearing as a variation in the complex phase.
The sign of the ${\rm Pf} D_X$ is, thus,
the only possible source of the anomaly, which is essentially nonperturbative.

Let us consider two gauge equivalent configurations
$A$ and $A^g$
smoothly connected by a one-parameter family,
say, parameterized by $s$: $A^s = (1-s)A + s A^g$.
Here, the configuration $A^g$ is obtained from $A$
by a $SU(2)$ gauge transformation $g$.
Since $A$ and $A^g$ are gauge equivalent,
exactly the same spectrum of the Dirac eigenvalues is shared.
However, some pairs of eigenvalues may be exchanged
crossing zero somewhere in $0<s<1$, which is called the spectral flow.
As ${\rm Pf} D_X$ is determined by only half of
the eigenvalue pairs, if this spectral flow is odd,
${\rm Pf} D_X$ changes its sign.

Identifying the infinity in $\mathbb{R}^4$ as one point,
or compactifying the spacetime to $S^4$,
the gauge transformations are classified by the homotopy group $\pi_4(SU(2))=\mathbb{Z}_2$.
In \cite{Witten:1982fp}, it was shown that 
when the gauge transformation is in the nontrivial class of $\pi_4(SU(2))$,
the eigenvalues must change the sign
by odd times and the spectral flow is odd.
Therefore, the sign of ${\rm Pf} D_X$ is not determined
in a gauge invariant way.

The proof was given using the mod-two Atiyah-Singer (AS) index. 
The one-parameter family $s$ given above can be treated as
the fifth dimension, to which the gauge connection $A^s$
is naturally introduced.
As the two points $s=0$ and $s=1$ are gauge equivalent,
the extended spacetime we consider is equivalent to
$S^4\times S^1$, which is called a mapping torus.
On this mapping torus, the Dirac operator $D$ is still real,
and the number of zero modes mod 2 is
known as the mod-two AS index.
It was proved that the mod-two AS index always agrees with
the spectral flow of original four-dimensional $D_X$ as follows.
Let us introduce another one-parameter family $t$,
which connects $A^0$ at $t=-\infty$ and $A^1$ at $t=\infty$, where
the $t$ dependence is mild.
The zero modes of $D$ satisfies
\begin{equation}
\label{eq:eigeneq}
\partial_{t}\Psi(t,x)=-\gamma^{t}D_X(t)\Psi(t,x),
\end{equation}
where $\gamma^{t}$ is the gamma matrix in the $t$-direction.
In this adiabatic situation, 
the solution is approximated by $\Psi(t,x)=\phi(t)\psi_{t}(x)$,
where $\psi_{t}(x)$ satisfies the four dimensional Dirac equation
$\gamma^{t}D_X(t)\psi_{t}(x)=\lambda(t)\psi_{t}(x)$, with the eigenvalue
$i\lambda(t)$ of $D_X(t)$ at the time slice $t$.
The solution $\phi(t)$ is formally given by
\begin{equation}
  \phi(t)=\phi(0)\exp\left[-\int_{0}^{t}dt^{\prime}\lambda (t^{\prime}) \right],
\end{equation}
but this is normalizable only when $\lambda(t) >0\  \text{for}\  t \to \infty$
and $\lambda(t) <0\  \text{for}\  t \to -\infty$.
Therefore, the number of zero modes of $iD$
always agrees with the spectral flow of $D_X$.
It was also shown that the index is always odd for $A^t$ when
the gauge transformation $g$ is in the nontrivial class of $\pi_4(SU(2))$.

This standard argument of global anomaly is
similar to the anomaly inflow of the perturbative anomalies
in that the extra dimension and associated Dirac operator are introduced.
However, as the extra direction introduced is limited to $S^1$,
it is difficult to treat the
original Weyl fermion as the edge localized mode of the total system.
The physical role of the bulk massive fermion is not obvious, either.
In fact, in the next subsection,
the notion of global anomaly is extended to incorporate
general bulk manifold attached to the original spacetime.
In the mathematical language, the extension
is from the mod-two AS index to the mod-two APS index.


\subsection{Global anomaly inflow (from mod-two AS to mod-two APS)}

To understand the anomaly inflow, it is instructive to
go back to the perturbative anomaly.
It is well-known that a single Weyl fermion in a complex
representation of $SU(N)$ ($N>2$) gauge interactions
suffers from anomaly and the theory is inconsistent.
However, the anomaly is exactly the same as
the surface term of the variation of the
Chern-Simons (CS) action and therefore,
the gauge invariance can be recovered by adding
a five-dimensional bulk fermion whose effective action
contains the CS action to cancel the anomaly of the Weyl fermion.
In this extension, known as the Callan-Harvey mechanism \cite{Callan:1984sa}
the original anomaly can be regarded
as a current escaping into the extra dimension, which is, in total,
conserved in the five-dimensional system.

The extended theory is still ``anomalous'' since the
theory is no more defined on the original four-dimensional manifold.
The theory is anomaly free when (the total sum of) CS action is zero.

In \cite{Witten:2015aba}, it was argued that
the Callan-Harvey mechanism can be applied to the global anomaly, as well.
The anomalous $n$-dimensional fermion path integral can be cured
by extending the theory to $(n+1)$-dimensions where
the total phase is given by $\exp(i\pi \eta(iD))$,
where $\eta(iD)$ is the $\eta$ invariant of the Dirac operator $iD$,
on the extended manifold \cite{Dai:1994kq,Freed:2014iua}.
Here, the $\eta$ invariant of a Hermitian operator $H$ is given by
a regularized sum of the sign of all eigenvalues $\lambda$,
\begin{equation}
  \eta(H) = \sum_\lambda \text{sgn} \lambda + h,
\end{equation}
where $h$ is the number of zero modes
(namely, we count the zero modes as positive eigenvalues.).
As the CS action is a perturbative part of the
$\eta$ invariant, the perturbative anomaly is
properly included in this anomaly inflow argument.

The $\eta$ invariant is gauge invariant and the total theory
is, thus, gauge invariant.
The theory is still ``anomalous'', since the
theory is no more defined on the original $n$-dimensional manifold $X$,
but depends on the extended $(n+1)$-dimensional bulk.
The theory is anomaly-free or consistent as a $n$-dimensional
theory, only when the (total) $\eta$ invariant
is independent of the bulk. 
Using the gluing property of the $\eta$ invariant,
this anomaly-free condition is simply given by $\eta=0$ (mod $2$)
on any closed manifold which is constructed by
gluing two $(n+1)$-dimensional manifolds sharing $X$,
the same $n$-dimensional boundary. 

When $D$ is real,
the $\eta$ invariant is reduced to
the number of zero modes, $h$
(Remember that all non-zero modes have $\pm$ pairs.).
Namely, this index is the mod-two APS index on a
$(n+1)$-dimensional manifold with the $n$-dimensional boundary.
The notion of anomaly is extended in that
we can put any $(n+1)$-dimensional bulk,
in contrast to the traditional global anomaly limited to
the mapping torus\footnote{
  Even in the framework of the mapping  torus,
  a new-type of anomaly in the 4-dimensional
  $SU(2)$ gauge theory was found \cite{Wang:2018qoy}.
}.

As is the case with perturbative anomaly,
if we can relate $\exp(i\pi \eta(iD))$ to
the path-integral of the massive fermion,
we may be able to unite the notion of anomaly as
the symmetry breaking of the $n$-dimensional
massless edge modes, which is cancelled by
the bulk massive fermions.
However, this is not straightforward since
the definition of $\eta(iD)$ requires
a special type of boundary condition,
known as the APS boundary condition,
to guarantee the Hermiticity of $iD$.



\subsection{Non-local boundary condition}

In the previous subsection, we have introduced
$\eta(iD)$, which describes
the phase of the fermion path-integral
in $(n+1)$-dimension in a gauge invariant way.
When $iD$ acts on a field in a real representation,
the mod-two APS index appeared as a special case
of the $\eta$ invariant whose non-zero eigenvalues cancel out.
But we have not discussed in detail what kind of
boundary conditions should be imposed on the $(n+1)$-dimensional fermions.

Before going into details, let us discuss yet another
fermion species, or those in a pseudo-real representation under
$Spin(n)$ and other symmetry group transformations.
This pseudo-real fermion is special in that it allows the mass term.
Therefore, any kind of gauge anomaly can be essentially removed by
the Pauli-Villars regularization, for example.
However, if we ``require'' the time-reversal ($T$) symmetry,
which is known to be incompatible with gauge symmetry in odd dimensions
and for odd number of Dirac fermions,
the situation is exactly the same as the previous complex and real fermion examples.
The gauge invariance needs bulk fermions.

In the pseudo-real fermion case, it is more natural
to consider the anomaly inflow as the one for $T$ symmetry, rather than gauge anomaly.
Let us consider a three-dimensional manifold $X$
and massless Dirac fermion with the $SU(N)(N>2)$ gauge interaction
on it, as an example,
\begin{eqnarray}
  \lim_{\Lambda \to \infty}\det \frac{D_X}{D_X+\Lambda}
  = \lim_{\Lambda \to \infty}\prod_\lambda \frac{i\lambda}{i\lambda+\Lambda}
  \propto \exp\left[-i\frac{\pi}{2}\eta(iD_X)\right],
\end{eqnarray}
where we have employed the Pauli-Villars regularization and
$i\lambda$ denotes the eigenvalue of $D_X$.
The $\eta$ invariant appears since the
phase of the determinant is essentially
given by how many times $i$ and $-i$ are multiplied,
which correspond to the number of positive $\lambda$
and negative $\lambda$, respectively.
The $T$ symmetry is broken as the $T$
transformation flips the sign of the mass $\Lambda$,
and thus the sign of $\eta(iD_X)$.

It is known that the smooth part
of $\eta(iD_X)$ is the Chern-Simons action,
half of which coincides the surface term
of the instanton number density integrated over
a four-dimensional manifold $Y$, whose boundary
is the original three-dimensional manifold $X$.
Thus we can add the bulk fermion so that
its effective action becomes this instanton number density.
The total phase
\begin{equation}
  \exp\left[i\pi
    \left\{
    P
    -\frac{1}{2}\eta(iD_X)\right\}\right]
   =\exp(i\pi \mathcal{I}),
\end{equation}
where $P$ an integral of local function of curvature\footnote{
  In four-dimensional flat space, it is well-known that
  $P=\frac{1}{32\pi^2}\int_Y d^4x \epsilon^{\mu\nu\sigma\rho}{\rm Tr}F_{\mu\nu}F_{\nu\rho}$.
  } over $Y$, 
is now guaranteed to be $T$ invariant,
as $\mathcal{I}$ is an integer known as the APS index.
The APS index $\mathcal{I}=n_+-n_-$ is defined by the
number of zeromodes $n_\pm$ with positive/negative chirality, respectively,
of the Dirac operator $iD$ on $Y$.
This index is again a special case of 
$\eta(iD)=h=n_++n_-=\mathcal{I}+2n_-$,
where $2n_-$ is irrelevant to the fermion determinant phase.
Note that the non-zero modes of $iD$ make $\pm$ pairs
by the chirality operator and do not contribute.
This APS index beautifully explains the bulk-edge correspondence
of the topological insulator where the $T$ symmetry is protected
by cancellation of the $T$ anomaly.

Now let us go back to the boundary condition of $D$.
For a complete set $\{\phi_i\}$ of the operand of $D$,
a natural choice would be $\gamma_\tau \phi_i|_X := n^\mu\gamma_\mu\psi_i|_X =\pm \phi_i|_X$,
where $n^\mu$ is a normal vector to the surface $X$.
This condition is local and respects rotational symmetry of $X$ when it exists.
However, this condition spoils the anti-Hermiticity of $D$
by the surface contribution as
\begin{eqnarray}
  \label{eq:antiH}
\int_Y d^{n+1}x \varphi_2^\dagger(x)D\varphi_1(x) + \int_Y d^{n+1}x (D\varphi_2(x))^\dagger\varphi_1(x) 
= \int_X d^{n}x \varphi_2^\dagger(x)\gamma_\tau \varphi_1(x),
\nonumber\\
\end{eqnarray}
for general $\varphi_1(x), \varphi_2(x)$ satisfying the same boundary condition.

Instead, the original work by APS \cite{APS} chooses
a different boundary condition, known as the APS boundary condition.
Assuming a product structure in the metric near the boundary,
and denoting the Dirac operator as $D= \gamma_\tau (\partial^\tau + A)$,
they require the boundary modes to satisfy\footnote{A more mathematical definition will be given in Sec.~\ref{sec:main}}
\begin{eqnarray}
\frac{A+|A|}{|A|}\varphi_i(x)|_{X} = 0. 
\end{eqnarray}
As $A$ anticommutes with $\gamma_\tau$, the surface contribution
in Eq.~(\ref{eq:antiH}) disappears to keep the anti-Hermiticity of $D$.
Moreover, $A$ commutes with the chirality operator and
the index of $D$ in terms of the chiral zeromodes is well-defined.
In \cite{Dai:1994kq,Freed:2014iua} a modified version was used,
but the essential properties of APS are inherited.

However, as discussed in details in \cite{Fukaya:2017tsq},
the APS boundary condition is unnatural and
unlikely to be realized in the materials.
In particular, the boundary condition has little relation to
the physics of topological insulators.
Let us examine if the edge localized solution $\exp(-\lambda \tau)$
can exist near the boundary.
If $\lambda$ is an eigenvalue of $A$, the Dirac equation holds.
But the solution is  normalizable only when
$\lambda$ is positive, which is not allowed by the APS boundary condition.
Namely, the APS condition prohibits the edge-localized modes to exist.
This makes it difficult to understand the bulk-edge correspondence
or anomaly inflow, in particular, in the mod-two APS index,
as it has no intuitive separation of the bulk and edge contributions,
in contrast to the $T$ anomaly inflow in the standard APS index,
or the perturbative gauge anomaly inflow of complex fermions.
It is also unnatural to lose the rotational symmetry
at the surface due to the gauge field dependence of $A$.
Above all, the operator $|A|$ is non-local, which
makes the causal structure of the system doubtful.

In \cite{Witten:2019bou}, Witten and Yonekura
explained how these unphysical properties of the APS boundary condition
are harmless when we consider the topological phase of materials.
They rotated the normal direction to the ``time'',
and treated the APS condition as an intermediate ``state''.
If the gap of the system is big enough,
the overlap between the physical boundary state
and the APS state is controlled by the ground state of the system
and the unphysical features of the APS cancel out
between the bra and ket states.
Their argument justifies
the use of the APS boundary condition in the physical system.
However, the fundamental question why the index or $\eta$ invariant
with such a unphysical property appears in the materials
is not clear.

A natural question is whether we can reformulate the APS index, or
$\eta$ invariant without relying on the non-local boundary condition.
To this question, a positive answer was partly
given in our previous works.
The key idea is to consider the so-called domain-wall fermion,
as discussed in the next section.



\section{Domain-wall fermion and standard APS index}
\label{eq:domainwall}

In the original work by Callan and Harvey \cite{Callan:1984sa},
where the anomaly inflow was first discussed,
they considered a spacetime $Y$ without any boundary,
rather than a manifold with boundary. 
Instead, they introduced a space-dependent fermion mass
(as a vacuum expectation value of scalar field) 
whose sign flips at some co-dimension one manifold $X$,
which divides $Y$ into $Y_+ \cup Y_-$.
Here $Y_\pm$ denotes the region with positive/negative fermion mass.
This is the so-called domain-wall fermion.
As we will see below, the domain-wall fermion is a good model to
describe the physics of topological insulators.
The region $Y_-$ corresponds to inside of topological insulator
and $Y_+$ is normal insulator.
This setup is more realistic than a manifold with boundary,
since any boundary in our world has ``outside'' of it.

Let us assume that $Y$ is an odd-dimensional manifold,
and $X$ is located at $\tau=0$ with
a simple product structure of the metric of $Y$ near $X$. 
Then the Dirac equation becomes
\begin{eqnarray}
0 = (D+m\kappa)\psi = \left(\gamma_\tau\partial_\tau + D_X + m\kappa \right)\psi =0,
\end{eqnarray}
where $\kappa={\rm sgn}(\tau)$ is a sign function such that
${\rm sgn}(\pm t)=\pm 1 $ for $t>0$, $D_X$ is the Dirac operator on $X$,
and $m>0$ is a real constant.
At  the leading order of adiabatic approximation
assuming slow $\tau$ dependence of the gauge field,
the above equation has an edge-localized solution \cite{Jackiw:1975fn}:
\begin{eqnarray}
\psi(x,\tau) = \phi(x)\exp(-m|\tau|),\;\;\; \gamma_\tau \phi(x)=\phi(x),\;\;\; D_X\phi(x)=0,
\end{eqnarray}
where $x$ is a local coordinate of $X$.
The last two conditions show that the edge mode has positive chirality,
and satisfies the massless Dirac equation on the domain-wall $X$.

In \cite{Callan:1984sa}, it was shown that
the edge-localized modes suffer from gauge anomaly,
but it is precisely canceled by the surface term of
the Chern-Simons action appearing as an effective action 
of the massive bulk modes in the region $Y_-$.
As the total massive Dirac fermion determinant in $Y$
can be regularized in a gauge invariant way, with Pauli-Villars fields, for instance,
this anomaly cancellation is guaranteed at all order of perturbation.
See \cite{Fukaya:2020trp} for a recent recomputation of this anomaly cancellation
in a more microscopic and subtle treatment of edge and bulk modes near the domain-wall.

In our recent work \cite{Fukaya:2017tsq,Fukaya:2019qlf}, we have successfully
described the anomaly inflow using the domain-wall fermion
when $Y$ is an even-dimensional manifold.
Let us consider a determinant of the domain-wall fermion
with Pauli-Villars regulator
\begin{eqnarray}
  \label{eq:DWdet}
\frac{\det (D+\kappa m)}{\det (D+m)} = \frac{\det i\gamma(D+\kappa m)}{\det i\gamma(D+m)},
\end{eqnarray}
where we have taken the physical mass and the Pauli-Villars
mass the same value for simplicity.
The sign function $\kappa$ again takes $\pm 1$ on $Y_\pm$.
Thanks to the existence of the chirality operator $\gamma$,
the determinant is always real since  $\det (D+\kappa m)=\det \gamma ( D+\kappa m)\gamma=\det ( D^\dagger+\kappa m).$

From the right-hand side of Eq.~(\ref{eq:DWdet}), one can see that
the sign of the determinant is controlled by the $\eta$ invariant of
the Hermitian operators $\gamma ( D+\kappa m)$ and $\gamma ( D+m)$.
And it coincides with the APS index ${\rm Ind}_{\rm APS} D$,
on the half of the manifold $Y_-$ with the APS boundary condition is imposed on $X$.
Namely, we have
\begin{eqnarray}
  {\rm Ind}_{\rm APS} D|_{Y_-} = -\frac{1}{2}\eta(\gamma ( D+\kappa m))+\frac{1}{2}\eta(\gamma ( D+m)).
\end{eqnarray}
This nontrivial equivalence  was perturbatively shown by three of the present authors \cite{Fukaya:2017tsq}.
Then the other three of the present authors who are mathematicians joined
the collaboration and gave a mathematical proof \cite{Fukaya:2019qlf}
that the agreement is not a coincidence but generally true 
on any even-dimensional curved manifold when $m$ is large enough.

In our reformulation of the APS index,
we put the Dirac operator on a closed even-dimensional
manifold $Y$, which ensures the anti-Hermiticity of $D$,
and no specific boundary condition is needed.
Instead, the local and rotational symmetric
boundary condition is automatically given on the domain-wall.
We have shown that the boundary $\eta$ invariant
entirely comes from the edge-modes localized on the wall,
and the curvature integral part in the index is from the bulk modes.
Thus, the bulk-edge correspondence is manifest in our reformulation.
The non-local feature of the boundary $\eta$ invariant
is also naturally explained by the masslessness of the edge modes.
This formulation is so physicist-friendly that
even the application to the lattice gauge theory is achieved \cite{Fukaya:2019myi}.

In this work, we pursue the mod-two version
of APS index, which applies to the real fermions in odd dimensions.
The most general case with complex fermions
is still an open question, although
we expect a similar relation between
the domain-wall fermion and $\eta(D)$ with the APS boundary condition.

\section{Main theorem}
\label{sec:main}
Here we describe our main theorem using a rather mathematically precise language. 
The physics consequence is discussed in the next section.

\subsection{Mod-two APS indices}\label{subsec_APS_def}
In this subsection, 
for completeness, we will define the mod-two APS index for real skew-adjoint operators on manifolds with boundaries, which is a slight modification of the original APS index \cite{APS} for self-adjoint operators on manifolds with boundaries\footnote{
The definition of the mod-two APS index is easy, but it is important that a mod-two version of the APS index {\it theorem} does not exist.  
}. 
$\mathcal{H}_\R$ denotes a separable real Hilbert space, and $\mathcal{H}_\C := \mathcal{H}_\R \otimes \C$ its complexification.  
A $\C$-linear operator $D$ on $\mathcal{H}_\C$ is called real if it coincides with its complex conjugate, and skew-adjoint if $D^\dagger = -D$. 
The complexification of an $\R$-linear operator $D$ on $\mathcal{H}_\R$ is also denoted by $D$, which is a real operator. 
The spectrum of an $\R$-linear operator on $\mathcal{H}_\R$ is always understood to be the spectrum of its complexification. 

Recall that, for a real skew-adjoint Fredholm operator $D$ on a separable real Hilbert space, the dimension mod $2$ of its kernel is a deformation invariant \cite{AS69}. 
So we define its {\it mod-two index} by
\begin{align*}
    \mathrm{Ind}(D) := \dim \ker D \pmod 2. 
\end{align*}
For a closed manifold equipped with a real vector bundle, the mod-two index of a skew-adjoint elliptic operator is defined in the above way and studied by the mod-two index theorem of Atiyah and Singer \cite{AS71}. 
Here we would like to formulate the mod-two APS index for the case with boundaries. 

Let $Y_-$ be a compact Riemannian manifold with boundary $X = \del Y_-$, and $S$ be a real Euclidean vector bundle over $Y_-$. 
We assume the collar structure $(-\epsilon, 0]_\tau \times X$ near the boundary of $Y_-$, and there exists a real Euclidean vector bundle $S_X$ over $X$ with the isomorphism $S_X \simeq S$ over the collar. 
We assume that $S_X$ is equipped with a self-adjoint endomorphism $\gamma_X \in \mathrm{End}(S_X)$ with $\gamma_X^2 = \mathrm{id}_{S_X}$. 
Let $D$ be a $\R$-linear formally skew-adjoint elliptic operator on $C^\infty(Y_-; S)$. 
Assume that, on the collar, $D$ is of the form
\begin{align*}
    D = \gamma_X \del_\tau + D_X, 
\end{align*}
for some $\R$-linear skew-adjoint elliptic operator $D_X$ on $C^\infty(X; S_X )$ which anti-commutes with $\gamma_X$, i.e., $\gamma_X D_X + D_X \gamma_X = 0$. 
In order to define the mod-two APS index, we assume that $D_X$ {\it is invertible}\footnote{The APS boundary condition is defined also in the case where $D_X$ has a nontrivial kernel, but the resulting operator is not skew-adjoint. }. 

In this setting, the APS boundary condition defined in \cite{APS} is the following.  
Note that $\gamma_X D_X$ is self-adjoint on $L^2(X; S_X)$. 
Let $P:= \chi_{[0, \infty)}(\gamma_X D_X)$ denote the spectral projection onto the non-negative eigenspaces of $\gamma_XD_X$. 
\begin{definition}[{the APS boundary condition (\cite{APS}) and mod-two APS indices}]
In the above settings, a smooth section $f \in C^\infty(Y_-; S)$ satisfies the APS boundary condition if it satisfies
\begin{align*}
    Pf|_{X} = 0. 
\end{align*}
The closure of this operator on $L^2(Y_-; S)$ with the above boundary condition, still denoted by $D$, is Fredholm. 
Moreover, if $D_X$ is invertible, $D$ is skew-adjoint. 
We define the mod-two APS index $\mathrm{Ind}_{\rm APS}(D) \in \Z_2$ of $D$ by its mod-two index. 
\end{definition}
The mod-two APS indices have another formulation as follows. 
We consider $Y_{\rm cyl}:= Y_- \cup [0,+\infty)\times X $
with the standard cylindrical-end metric.
The bundle $S$ and the operator $D$ naturally extend to
$Y_{\rm cyl}$, which is denoted by $S_{\rm cyl}$ and $D_{\rm cyl}$.
\begin{proposition}[{\cite[Proposition 3.11]{APS}}]\label{prop_APS=cyl}
If $D_X$ is invertible, $D_{\rm cyl}$ is a skew-adjoint Fredholm operator on $L^2(Y_{\rm cyl}; S_{\rm cyl})$. 
Let us denote by $\mathrm{Ind}(D_{\rm cyl})$ its mod-two index. 
We have
\begin{align*}
    \mathrm{Ind}_{\rm APS}(D) = \mathrm{Ind}(D_{\rm cyl}). 
\end{align*}
\end{proposition}

\subsection{The statement of the main theorem}

Let $Y$ be a closed Riemannian manifold of which dimension can be odd or even.
Let $S$ be a real Euclidean vector bundle on $Y$. 
Let $D:C^\infty (Y;S)\to C^\infty (Y;S)$ be a first-order, formally skew-adjoint,
elliptic partial differential operator.
Let $X\subset Y$ be a separating submanifold that decomposes $Y$ into two compact manifolds
$Y_+$ and $Y_-$ with common boundary $X$.
Let $\kappa: Y \to [-1,1]$ be the $L^\infty$--function such that $\kappa\equiv \pm 1$
on $Y_\pm \backslash X$.
We define $D_{\rm DW}=D+\kappa m {\rm id}_S$ with a real positive number $m$
as a domain-wall Dirac operator, where ${\rm id}_S$ is an identity matrix on $S$.
We also define $D_{\rm PV}=D + m {\rm id}_S$, of which determinant
corresponds to the Pauli-Villars regulator.

We assume that $X$ has a collar neighborhood isometric to the standard
product $(-4,4)\times X$ and satisfying $((-4,4)\times X) \cap Y_- = (-4,0]\times X$.
  We denote the coordinate along $(-4,4)$ by $\tau$.
  We assume the collar structure on $S$ and $D$ explained in subsection \ref{subsec_APS_def}. 
  

In the collar region, $D_{\rm DW}$ can be written as
\begin{equation}
 D_{\rm DW}= \gamma_X(\partial_X + \gamma_X \kappa m {\rm id}_S + \gamma_X D_X ). 
\end{equation}
For $m$ large enough, $D_{\rm DW}$ is invertible. 
This can be shown in the same way as \cite[Proposition 9]{Fukaya:2019qlf}, and can be understood as follows. 
In the large $m$ limit, we have edge-localized modes proportional to
$\exp(-m|\tau|)$ in a $\gamma_X=+1$ subspace, on which
the domain-wall Dirac operator  operates as $D_{\rm DW}=DP_+$,
where $P_+=(1+\gamma_X)/2$. When $D_X$ at $\tau=0$ has no zero eigenvalue,
it is guaranteed that $D_{\rm DW}$ is invertible.

\begin{theorem}\label{theorem: main theorem}
  If $D_X$ on $C^\infty(X; S_X) $ is invertible, 
  then there exists a constant $m_0 > 0$ that depends only on $X$, $S$, and $D$ such that
  for any $m > m_0$ we have,
  \begin{equation}
    \label{eq:maintheorem}
    {\rm Ind}_{\rm APS}(D|_{Y_-}) = \frac{1-{\rm sgn}\det(D_{\rm DW}D_{\rm PV}^{-1})}{2} \pmod{2}.
  \end{equation}
  Were, 
  ${\rm sgn}\det(D_{\rm DW}D_{\rm PV}^{-1})$ will be defined in Definition \ref{def_signdet} below. 
\end{theorem}

Here, ``$\rm{sgn} \det$'' in the right hand side of \eqref{eq:maintheorem} needs an explanation, because the operator $D_{\rm DW}D_{\rm PV}^{-1}$ is defined on infinite-dimensional Hilbert space. 
Note that the real invertible operator $D_{\rm DW}D_{\rm PV}^{-1}$ differs from the identity operator by a compact operator. 
For such operators, we define ``$\rm{sgn} \det$'' which generalizes the usual signature of the determinants of invertible real operators on finite-dimensional Hilbert spaces as follows. 
For a real Hilbert space $\mathcal{H}_\R$, let 
\begin{align}
    \mathcal{C}(\mathcal{H}_\R) := \{A \in {\rm Id}_{\mathcal{H}_\R} + \mathcal{K}(\mathcal{H}_\R) \ | \ 
    A \mbox{ is invertible }\},
\end{align}
where $\mathcal{K}(\mathcal{H}_\R) $ denotes the space of compact operators on $\mathcal{H}_\R$. 
The space $\mathcal{C}(\mathcal{H}_\R)$, equipped with the norm topology, has two connected components \cite[Proposition 3.3]{AS69}. 

\begin{definition}[{${\rm sgn}\det$}]\label{def_signdet}
We define a map
\begin{align*}
    {\rm sgn}\det \colon \mathcal{C}(\mathcal{H}_\R) \to \{1, -1\}
\end{align*}
by letting ${\rm sgn}\det(A) := 1$ if $A$ belongs to the same connected component of $\mathcal{C}(\mathcal{H}_\R)$ with the identity, and ${\rm sgn}\det(A) := -1$ otherwise. 
\end{definition}
This map is a generalization of the ``$\rm{sgn} \det$'' for finite-dimensional case. 
Indeed, if $A \in \mathcal{C}(\mathcal{H}_\R)$ is of the form
$A = A_V \oplus \mathrm{id}_{V^\perp}$ for some finite dimensional subspace $V \subset \mathcal{H}_\R$, 
then the value ${\rm sgn}\det(A)$ defined in Definition \ref{def_signdet} coincides with the signature of the determinant of $A_V$.

\subsection{Example on a closed manifold}

Before giving a general proof, let us  consider  a special case
with $Y_-=Y$ or $\kappa=-1$ on whole $Y$
and there is no domain-wall.
In this case, we obtain the mod-two AS index on whole $Y$.
\begin{corollary}\label{col:AStheorem}
For any $m > 0$, we have
  \begin{equation}
    {\rm Ind}_{\rm AS}(D) = \frac{1-{\rm sgn}\det\left[(D-m{\rm id}_S  )(D+m{\rm id}_S )^{-1}\right]}{2} \pmod{2},
  \end{equation}
  where ${\rm Ind}_{\rm AS}(D)=\dim \ker(D) \pmod{2}$.
\end{corollary}

This corollary can be easily checked by a direct evaluation
of the massive fermion determinant.
Remembering that every non-zero eigenvalue $i\lambda$ of $D$
makes a pair with another eigenvalue $-i\lambda$ (where $\lambda$ is real)
the ratio of the determinant is expressed as
\begin{align*}
    \det\left[(D-m{\rm id}_S  )(D+m{\rm id}_S )^{-1}\right]
  &= \frac{(-m)^{N_0}\prod_{\lambda>0}(\lambda^2+m^2)}{m^{N_0}\prod_{\lambda>0}(\lambda^2+m^2)}\\
  &= (-1)^{N_0},
\end{align*}
where $N_0$ is the number of zero modes, or $N_0=Ind_{\rm AS}(D)$ mod 2.

\subsection{Mathematical preparations: mod-two spectral flows and indices on cylinders}
In this subsection, we give mathematical preparations necessary for the proof of the main theorem. 
In  \cite{Carey},  Carey, Phillips and Schulz-Baldes introduced mod-two spectral flow for paths of real skew-adjoint Fredholm operators. 
After recalling it and its necessary properties, we relate it with ``$\mathrm{sgn}\det$'' in Definition \ref{def_signdet}. 
We also explain its relation with mod-two indices of operators on cylinders. 

We have to deal with unbounded operators on Hilbert spaces. 
Unbounded operators appearing below are always assumed to be closed and densely defined. 
We topologize the set of unbounded closed densely defined Fredholm operators by the Riesz topology (see for example \cite{Nicolaescu} for topologies on the space of unbounded Fredholm operators). 
In this topology, a family $\{D_x\}_{x \in X}$ of Fredholm operators is continuous if and only if the families $\{D_x(1 + D^\dagger_x D_x)^{-1/2}\}_{x \in X}$ and $\{D^\dagger_x(1 + D_x D_x^\dagger)^{-1/2}\}$ are both continuous with respect to the norm topology. 
Restricted to the subspace of bounded Fredholm operators, it coincides with the norm topology. 
 
Now, we recall the definition of mod-two spectral flows for continuous paths of real skew-adjoint operators following \cite{Carey}. 
The spectrum of a real skew-adjoint operator $D$ lies in $\sqrt{-1}\R$, and satisfies $\mathrm{Spec}(D) = -\mathrm{Spec}(D)$. 
In generic cases, the mod-two spectral flow counts the parity of the number of changes in the orientation of the eigenfunctions at eigenvalue crossings through $0$ along the path. 

First assume that $\mathcal{H}_\R$ is finite dimensional. 
Given two invertible real skew-adjoint operators $D_{-1}$ and $D_1$, the mod-two spectral flow between them is defined as follows. 
Choose an operator $A$ on $\mathcal{H}_\R$ such that
\[
D_{1} = A^\dagger D_{-1} A. 
\]
Then the mod-two spectral flow in the finite-dimensional case is simply
\begin{align}\label{eq_fin_dim_sf}
    \mathrm{Sf}(D_{-1}, D_1) := \frac{1 - {\rm sgn}\det(A)}{2} \in \Z_2. 
\end{align}

Next we pass to the infinite-dimensional case. 
The definition in \cite{Carey} is given for bounded families, but it is straightforward to extend it to the unbounded case\footnote{
In \cite{Carey20} the authors extend the definition of spectral flows to paths of operators with general Clifford symmetries. 
There, they also define spectral flows for paths of unbounded Fredholm operators. }. 
The precise definition of mod-two spectral flow consists of subdividing a path into pieces small enough, and applying the definition for finite-dimensional paths for each pieces. 
Assume we are given a continuous family $\{D_t\}_{t \in [a, b]}$ of real skew-adjoint Fredholm operators on $\mathcal{H}_\R$, parameterized by a finite interval $[a, b] \subset \R$. 
We assume that $D_a$ and $D_b$ are invertible. 
For $\lambda >0$ and $t \in [a,b]$, we define the corresponding spectral projection by
\begin{align*}
    Q_\lambda(t) := \chi_{(-\lambda, \lambda)} (\sqrt{-1}D_t),  
\end{align*}
where $\chi_{(-\lambda, \lambda)}$ is the characteristic function of $(-\lambda, \lambda)$. 
$Q_\lambda(t)$ is a real projection. 
By Fredholmness of $D_t$, for $\lambda$ small enough $Q_\lambda(t) \mathcal{H}_\R$ is a finite dimensional subspace of $\mathcal{H}_\R$. 
For each $t$, take an arbitrary skew-adjoint operator $R_t$ on the kernel of $Q_\lambda(t)D_tQ_\lambda(t)$. 
Let us denote
\begin{align*}
    D_t^{(\lambda)} :=  Q_\lambda(t)D_tQ_\lambda(t) + R_t. 
\end{align*}
This is a real skew-adjoint invertible operator on $Q_\lambda(t) \mathcal{H}_\R$. 

We choose a subdivision of the interval as $a = t_0 < t_1 < \cdots < t_N = b$,
and a sequence of positive numbers $\{\lambda_n\}_{n=1}^N$ such that $Q_{\lambda_n}(t)$ is of constant finite rank on the interval $[t_{n - 1}, t_n]$ for all $n$, and
the orthogonal projection
\begin{align*}
    V_n \colon Q_{\lambda_n}(t_{n - 1}) \mathcal{H}_\R \to Q_{\lambda_n}(t_{n }) \mathcal{H}_\R 
\end{align*}
is a bijection for all $n$. 
Using these data, the mod-two spectral flow of the path $\{D_t\}_t$ is defined as follows. 

\begin{definition}[Mod-two spectral flows][{\cite[Definition 4.1]{Carey}}]\label{def_sf}
Let $\{D_t\}_{t \in [a, b]}$ be a continuous path of real skew-adjoint possibly unbounded Fredholm operators on $\mathcal{H}_\R$. 
We assume that $D_a$ and $D_b$ are both invertible. 
Choosing additional datum as above, we define the spectral flow of the path $\{D_t\}_{t \in [a, b]}$ by
\begin{align*}
    \mathrm{Sf}(\{D_t\}_t) := \sum_{n = 1}^N \mathrm{Sf}(D_{t_{n-1}}^{(\lambda_n)}, 
    V_n^\dagger D_{t_n}^{(\lambda_n)}V_n). 
\end{align*}
 
\end{definition}

For an unbounded path $\{D_t\}_{t \in [a, b]}$, we can also take the bounded transform
$\{D_t (1 + D_t^\dagger D_t)^{-1/2}\}_{t\in [a, b]}$ to get a bounded path, and consider its mod-two spectral flow. 
We easily see that
\begin{align}\label{eq_bdd_transf}
    \mathrm{Sf}(\{D_t\}_{t }) =  \mathrm{Sf}(\{D_t (1 +D^\dagger_t D_t)^{-1/2}\}_{t }). 
\end{align}

\subsubsection{The case of paths consisting of bounded operators}
In this subsubsection, we deal with paths consisting of bounded operators. 
We relate ``${\rm sgn} \det$'' in Definition \ref{def_signdet} with a certain type of mod-two spectral flows. 

In general, spectral flows are not determined by the operators at the endpoints, but depend on the choice of the paths. 
However, continuous deformations of the paths do not change the spectral flows, as long as they fix the endpoints (\cite[Theorem 4.3]{Carey}). 
This implies the following. 
\begin{lemma}\label{lem_cpt_perturbation}
Given two bounded paths $\{D_t\}_{t \in [a, b]}$ and $\{D'_t\}_{t \in [a, b]}$ satisfying the conditions in Definition \ref{def_sf}, 
assume $D_a = D'_a$, $D_b = D'_b$, and that $D_t - D'_t$ is a compact operator for all $t \in [a, b]$. 
Then we have
\begin{align*}
    \mathrm{Sf}(\{D_t\}_t) = \mathrm{Sf}(\{D'_t\}_t). 
\end{align*}
\end{lemma}

\begin{proof}
Since the Fredholmness is preserved by adding compact operators, we get a continuous deformation between two paths $\{D_t\}_t$ and $\{D'_t\}_t$ by the linear homotopy. 
\end{proof}

Thus, if we are given two invertible real skew-adjoint Fredholm operators $D_{-1}$ and $D_{1}$ which differ by a compact operator, we get a distinguished value of spectral flows between them; namely those of paths consisting of compact perturbations between them. 
\begin{definition}\label{def_sf_cpt}
Let $D_{-1}$ and $D_1$ be two invertible real skew-adjoint bounded Fredholm operators on $\mathcal{H}_\R$. 
Assume that $(D_1 - D_{-1})$ is a compact operator. 
Take any path $\{D_t\}_{t \in [-1, 1]}$ of real skew-adjoint Fredholm operators connecting $D_{-1}$ and $D_1$, such that $(D_t - D_{-1})$ is a compact operator for all $t \in [-1, 1]$. 
Then we define
\begin{align*}
    \mathrm{Sf}_{\rm cpt}(D_{-1}, D_1) := \mathrm{Sf}(\{D_t\}_{t \in [-1, 1]}). 
\end{align*}
This does not depend on the choice of the path by Lemma \ref{lem_cpt_perturbation}. 
\end{definition}

For $\mathrm{Sf}_{\rm cpt}$, we have a similar formula as \eqref{eq_fin_dim_sf}, which expresses the spectral flow by ``${\rm sgn}\det$'' of operators defined in Definition \ref{def_signdet}. 

\begin{proposition}\label{prop_sf=signdet}
Let $D_{-1}$ and $D_1$ be two invertible real skew-adjoint bounded operators on $\mathcal{H}_\R$. 
Assume that there exists an element $A \in \mathcal{C}(\mathcal{H}_\R)$ such that
\begin{align*}
    D_{1} =  A^\dagger D_{-1} A. 
\end{align*}
In particular, this means that $D_1 - D_{-1}$ is compact. 
Then, we have
\begin{align*}
    \mathrm{Sf}_{\rm cpt}(D_{-1}, D_1) := \frac{1 - {\rm sgn}\det(A)}{2},
\end{align*}
where ${\rm sgn}\det(A)$ is defined in Definition \ref{def_signdet}. 
\end{proposition}

\begin{proof}
Choose $\lambda > 0$ so that the spectrum of the spectrum of $\sqrt{-1}D_{-1}$ is discrete on the interval $[-\lambda, \lambda]$. 
The Hilbert space $\mathcal{H}_\R$ is decomposed as
\[\mathcal{H}_\R = Q_\lambda(-1) \mathcal{H}_\R \oplus (1 - Q_\lambda(-1)) \mathcal{H}_\R \]
with the first component finite dimensional. 
Choose a continuous path $\{A_t\}_{t \in [1, 2]}$ in $\mathcal{C}(\mathcal{H}_\R)$ such that $A = A_1$ and $A_2$ is of the form
\[
A_2 = A_2|_{Q_\lambda(-1) \mathcal{H}_\R} \oplus \mathrm{id}_{(1 - Q_\lambda(-1)) \mathcal{H}_\R}. 
\]
This implies
\begin{align}\label{eq_proof_signdet}
    {\rm sgn}\det(A) = {\rm sgn}\det(A_2|_{Q_\lambda(-1) \mathcal{H}_\R}). 
\end{align}
Consider a path $\{D_t\}_{t \in [-1, 2]}$ defined as 
\begin{align*}
    D_t := \begin{cases}
    \frac{1 - t}{2}D_{-1} + \frac{t + 1}{2}D_1 & \mbox{if } t \in [-1, 1],  \\
    A^\dagger_t D_{-1} A_t & \mbox{if } t \in [1, 2]. 
    \end{cases}
\end{align*}
Then we have
\begin{align*}
    \mathrm{Sf}_{\rm cpt}(D_{-1}, D_2) = \mathrm{Sf}(\{D_t\}_{t \in [-1, 1]}) + \mathrm{Sf}(\{D_t\}_{t \in [1, 2]})
    = \mathrm{Sf}_{\rm cpt}(D_{-1}, D_1),  
\end{align*}
where the second equality follows from the invertiblity of the family $\{D_t\}_{t\in [1, 2]}$. 
Note that $\mathrm{Sf}_{\rm cpt}(D_{-1}, D_2)$ is equal to the spectral flow of the linear path between $D_{-1}$ and $D_2$. 
Applying Definition \ref{def_sf} to this linear path, we see that
\begin{align*}
    \mathrm{Sf}_{\rm cpt}(D_{-1}, D_2) = \frac{1 - {\rm sgn}\det(A_2|_{Q_\lambda(-1) \mathcal{H}_\R})}{2}.
\end{align*}
Combining these with \eqref{eq_proof_signdet}, we get the result. 
\end{proof}

\subsubsection{The case of paths consisting of elliptic pseudodifferential operators}
In this subsubsection, we deal with the paths $\{D_t\}_t$ consisting of first order elliptic pseudodifferential operators on closed manifolds. 
Let us assume that $\mathcal{H}_\R = L^2(Y; S)$, where $Y$ is a closed manifold and $S$ is an $\R$-vector bundle with inner product over $Y$. 
Using the relation \eqref{eq_bdd_transf}, we have the corresponding notion of $\mathrm{Sf}_{\rm cpt}$ in this setting. 

\begin{definition}\label{def_sf_cpt_elliptic}
Let $Y$ and $S$ as above. 
Let $D_{-1}$ and $D_1$ be two invertible real skew-adjoint first order elliptic pseudodifferential operators on $L^2(Y; S)$. 
Assume that $D_1 - D_{-1}$ is of zeroth order. 
Take any path $\{D_t\}_{t \in [-1, 1]}$ of real skew-adjoint elliptic operators connecting $D_{-1}$ and $D_1$, such that $D_t - D_{-1}$ is of zeroth order for all $t \in [-1, 1]$. 
Then we define
\begin{align*}
    \mathrm{Sf}_{\rm cpt}(D_{-1}, D_1) := \mathrm{Sf}(\{D_t\}_{t \in [-1, 1]}). 
\end{align*}
This does not depend on the choice of the path by Lemma \ref{lem_cpt_perturbation} and \eqref{eq_bdd_transf}. 
\end{definition}

We see that $\mathrm{Sf}_{\rm cpt}$ is also compatible with the bounded transform, 
\begin{align}\label{eq_sf_cpt_bdd_transf}
    \mathrm{Sf}_{\rm cpt}(D_{-1}, D_1) = \mathrm{Sf}_{\rm cpt}(D_{-1}(1 + D^\dagger_{-1}D_{-1})^{-1/2}, D_{1}(1 + D^\dagger_{1}D_{1})^{-1/2}), 
\end{align}
where the left hand side is defined in Definition \ref{def_sf_cpt_elliptic} and the right hand side is defined in Definition \ref{def_sf_cpt}. 

\subsubsection{A relation between mod-two APS indices on cylinders and mod-two spectral flows}

In this subsection, we assume that $\mathcal{H}_\R$ is $\Z_2$-graded. 
Let $\gamma \in O(\mathcal{H}_\R)$ denote the $\Z_2$-grading operator. 
We deal with both cases where a family $\{D_t\}_{t}$ is bounded and unbounded. 
We explain a relation between mod two spectral flows of {\it odd} real skew-adjoint Fredholm operators and mod-two indices of certain operators on $\R$.  

\begin{proposition}\label{prop_APS=sf}
Let $\mathcal{H}_\R$ be $\Z_2$-graded with the grading operator $\gamma$. 
Let $\{D_t\}_{t \in [a, b]}$ be a continuous path of odd (i.e., $\gamma D_t + D_t \gamma = 0$ for all $t$) real skew-adjoint possibly unbounded Fredholm operators on $\mathcal{H}_\R$. 
We assume that $D_a$ and $D_b$ are both invertible. 

We construct a real skew-adjoint operator $\hat{D}$ on $L^2(\R_t) \otimes \mathcal{H}_\R$ as follows. 
By a continuous homotopy which fixes the endpoints, we perturb the path $\{D_t\}_{t \in [a, b]}$ into a smooth path $\{D^{\rm sm}_t\}_{t \in [a, b]}$ which is constant near the endpoints. 
We extend the path to $\{D^{\rm sm}_{t}\}_{t \in \R}$ by letting $D^{\rm sm}_t = D_a$ for $t < a$ and $D^{\rm sm}_t = D_b$ for $t > b$. 
We define $\hat{D}$ as
\begin{align*}
    \hat{D} := \gamma \del_t + D^{\rm sm}_t . 
\end{align*}
Then $\hat{D}$ is Fredholm and its mod-two index does not depend on the choice of the smoothing $\{D^{\rm sm}_t\}_{t \in [a, b]}$.  
We have, 
\begin{align}\label{eq_prop_APS=sf}
    \mathrm{Ind}(\hat{D}) = \mathrm{Sf}(\{D_t\}_{t \in [a,b]}) \in \Z_2. 
\end{align}
\end{proposition}
\begin{proof}
The independence of $\mathrm{Ind}(\hat{D})$ on the choice of smoothings follows from the deformation invariance of indices. 

We reduce the proof of \eqref{eq_prop_APS=sf} to finite-dimensional cases. 
In order to do this, we need the following easy properties of the indices of operators on cylinders.
\begin{enumerate}
    \item[(a)] Given a path $\{D_t\}_{t \in [a, b]}$ as above, if $D_t$ is invertible for all $t \in [a, b]$, we have $\mathrm{Ind}(\hat{D}) = 0$. 
    \item[(b)] Given a path $\{D_t\}_{t \in [a, b]}$ as above, assume that the path is divided into two paths as $\{D_t\}_{t \in [a, b]} = \{D'_t\}_{t \in [a, c]} \cup \{D''_t\}_{t \in [c, a]}$ with $D_c$ invertible. 
    We construct the operators $\hat{D'}$ and $\hat{D''}$ on $L^2(\R)\otimes \mathcal{H}_\R$ using $\{D'_t\}_t$ and $\{D''_t\}_t$ respectively in the same way. 
    Then we have
    \begin{align*}
        \mathrm{Ind}(\hat{D}) = \mathrm{Ind}(\hat{D'})+ \mathrm{Ind}(\hat{D''}) . 
    \end{align*}
    \item[(c)] Given two paths $\{D_t\}_{t \in [a, b]}$ and $\{D'_t\}_{t \in [a, b]}$ as above, assume that $D_a = D'_a$ and $D_b = D'_b$, and that the two paths are connected by a continuous homotopy leaving the endpoints fixed. 
    Then we have
    \begin{align*}
        \mathrm{Ind}(\hat{D}) = \mathrm{Ind}(\hat{D'}). 
    \end{align*}
\end{enumerate}
Indeed, (a) is because $\hat{D}$ is invertible in such cases, 
(b) follows from the gluing property of the indices, 
and (c) follows from the deformation invariance of the indices. 
Using the definition of mod-two spectral flows and the above properties of indices of operators on cylinders, as well as the corresponding properties of mod-two spectral flows (\cite[Theorem 4.2, 4.3]{Carey}), we can easily reduce to the case where $\mathcal{H}_\R$ is finite dimensional. 
Moreover, using the above properties again, we are reduced to the case where 
$\mathcal{H}_\R = \R^2$, 
\begin{align*}
    \gamma = \left(
    \begin{array}{cc}
      1 & 0 \\
      0 & -1
    \end{array}
  \right), \quad
  D_t = t\left(
    \begin{array}{cc}
      0 & 1 \\
      -1 & 0
    \end{array}
  \right) \quad t \in [-1, 1]. 
\end{align*}
In this case we have $\mathrm{Sf}(\{D_t\}_{t \in [-1, 1]}) = 1$. 
On the other hand, the $L^2$-kernel of $\hat{D}$ is one-dimensional, spanned by an element which is asymptotically $e^{-t}(1, -1)$ on $t >> 1$ and $e^{t}(1, -1)$ on $t << -1$, so we have $\mathrm{Ind}(\hat{D})=1$. 
Thus we get \eqref{eq_prop_APS=sf} and the result follows. 

\end{proof}

Now assume that we are given two invertible odd real skew-adjoint operators $D_{-1}$ and $D_1$ which differ by compact (resp. zero-th order) in the bounded case (resp. first-order elliptic case). 
In this case, we get a canonical choice of operator $A$ satisfying $D_{1} =  A^\dagger D_{-1} A$. 
Namely, with respect to the $\Z_2$-grading, we decompose $D_t$, $t = \pm 1$, as
\begin{align}\label{eq_odd_decomp}
    D_t = \left(
    \begin{array}{cc}
      0 & D_{+, t} \\
      -(D_{+, t})^\dagger & 0 
    \end{array}
  \right). 
\end{align}
Then we can choose $A$ to be, 
\begin{align*}
    A := \left(
    \begin{array}{cc}
      (D_{+, -1}^\dagger)^{-1}D_{+, 1}^\dagger & 0 \\
      0 & \mathrm{id} 
    \end{array}
  \right). 
\end{align*}
By the assumption on the difference between $D_1$ and $D_{-1}$, we see that $A \in \mathcal{C}(\mathcal{H}_\R)$. 
In the bounded case, by Proposition \ref{prop_sf=signdet} and the obvious identity ${\rm sgn} \det (A) = {\rm sgn} \det (A^\dagger) $, we get the following. 
\begin{proposition}\label{prop_sf=signdet_odd_bdd}
Let $\mathcal{H}_\R$ be $\Z_2$-graded with the grading operator $\gamma$. 
Assume we are given two invertible odd real skew-adjoint bounded operators $D_{-1}$ and $D_1$ with $D_1 - D_{-1}$ compact. 
Then we have
\begin{align*}
    \mathrm{Sf}_{\rm cpt}(D_{-1}, D_1) = \frac{1 - {\rm sgn}\det(D_{+, 1}(D_{+, -1})^{-1})}{2}. 
\end{align*}
Here $D_{+, t}$ is defined in \eqref{eq_odd_decomp}.  
\end{proposition}

In the first-order elliptic case, we have the corresponding result. 
\begin{proposition}\label{prop_sf=signdet_odd_elliptic}
Let $Y$ be a closed manifold and $S$ be a $\Z_2$-graded real Euclidean vector bundle over $Y$. 
Assume we are given two invertible odd real skew-adjoint first-order elliptic operators $D_{-1}$ and $D_1$ on $L^2(Y; S)$. 
Suppose that $D_1 - D_{-1}$ is of zeroth order. 
Then we have
\begin{align*}
    \mathrm{Sf}_{\rm cpt}(D_{-1}, D_1) = \frac{1 - {\rm sgn}\det(D_{+, 1}(D_{+, -1})^{-1})}{2}. 
\end{align*}
Here $D_{+, t}$ is defined in \eqref{eq_odd_decomp}. 
\end{proposition}
\begin{proof}
For an unbounded Fredholm operator, let us denote by $\chi(D) := D(1 + D^\dagger D)^{-1/2}$ its bounded transform.
Note that if $D$ is odd and skew-adjoint, $\chi(D)$ is also odd and skew-adjoint. 
By \eqref{eq_sf_cpt_bdd_transf} and Proposition \ref{prop_sf=signdet_odd_bdd}, we get
\begin{align*}
    \mathrm{Sf}_{\rm cpt}(D_{-1}, D_1) =\frac{1 - {\rm sgn}\det(\chi(D_{+, 1})(\chi(D_{+, -1}))^{-1})}{2}. 
\end{align*}
Since the operator $(1 + D_{+, 1}^\dagger D_{+, 1})^{1/2}(1 + D_{+, -1}^\dagger D_{+, -1})^{-1/2}$ lies in the same connected component of $\mathcal{C}(\mathcal{H}_\R)$ as the identity, we see that
\begin{align*}
{\rm sgn}\det(D_{+, 1}(D_{+, -1})^{-1}) = 
    {\rm sgn}\det(\chi(D_{+, 1})(\chi(D_{+, -1}))^{-1}). 
\end{align*}
Thus, we get the result. 
\end{proof}

\subsection{Proof of main theorem}\label{subsec_proof}
In this subsection, we prove Theorem \ref{theorem: main theorem}. 
The proof given here relies on the techniques developed in our previous work \cite{Fukaya:2019qlf}. 
We will see that, by modifying the proof in that paper appropriately, essentially the same proof works in the mod-two case. 
We give an alternative simpler and self-contained proof in Appendix. 

First, in order to deal with smooth operators in the proof, we perturb the $L^\infty$-function $\kappa \colon Y \to [-1, 1]$ to a smooth function $\kappa^{\rm sm} \colon Y \to [-1, 1]$ so that $\kappa^{\rm sm} \equiv \pm1$ on $Y_{\pm} \setminus (-4, 4) \times X$ (recall the collar parameter introduced before the statement of Theorem \ref{theorem: main theorem}). 
Consider the corresponding smoothed
domain-wall Dirac operator, 
\begin{equation*}
 D^{\rm sm}_{\rm DW}=  D + \kappa^{\rm sm} m {\rm id}_S.
\end{equation*}
For $m$ large enough, we have
\begin{align*}
    {\rm sgn}\det(D_{\rm DW}D_{\rm PV}^{-1}) = {\rm sgn}\det(D^{\rm sm}_{\rm DW}D_{\rm PV}^{-1}). 
\end{align*}
This is because, for $m$ large enough, the linear path connecting $D_{\rm DW}$ and $D_{\rm DW}^{\rm sm}$ consists of invertible operators. 

Let us consider a $\Z_2$-graded vector bundle $S \oplus S$ over $Y$ with the natural real structure, with the $\Z_2$-grading given by $\gamma = \mathrm{diag}(\mathrm{id}_{S},- \mathrm{id}_S)$. 
Choose any smooth function $\hat{\kappa}^{\rm sm} \colon \R \times Y \to [-1, 1]$
such that $\hat{\kappa}^{\rm sm}_t := \hat{\kappa}^{\rm sm}(t, \cdot) = +1$ for $t < -0.5$ and $\hat{\kappa}^{\rm sm}_t = \kappa^{\rm sm}$ for $t > 0.5$. 
Let $D_t\colon L^2(Y;S\oplus S)\to L^2 (Y;S\oplus S)$
be a one-parameter family of odd real skew-adjoint elliptic operators defined by
\begin{eqnarray}\label{eq_Dt}
  D_t := \left(
  \begin{array}{cc}
    0 & D+m\hat{\kappa}^{\rm sm}_t  {\rm id}_S\\
    D-m\hat{\kappa}^{\rm sm}_t  {\rm id}_S & 0
\end{array}
  \right). 
\end{eqnarray}  
Note that
\begin{eqnarray}
 D_{-1} = \left(
  \begin{array}{cc}
    0 & D_{\rm PV}\\
    -(D_{\rm PV})^\dagger & 0
\end{array}
  \right),&
  D_{1} = \left(
  \begin{array}{cc}
    0 & D^{\rm sm}_{\rm DW}\\
    -(D^{\rm sm}_{\rm DW})^\dagger & 0
\end{array}
  \right),& 
\end{eqnarray}
and these two operators are both invertible.
The path $\{D_t\}_{t \in [-1, 1]}$ satisfies the condition in Definition \ref{def_sf_cpt_elliptic}; namely, $D_t - D_{-1}$ is of zeroth-order for all $t \in [-1, 1]$. 
By Proposition \ref{prop_sf=signdet_odd_elliptic}, we get
\begin{align}\label{eq_prf_main_sf}
    \mathrm{Sf}(\{D_t\}_{t \in [-1, 1]}) = \mathrm{Sf}_{\rm cpt}(D_{-1}, D_1) 
    = \frac{1-{\rm sgn}\det(D^{\rm sm}_{\rm DW}D_{\rm PV}^{-1})}{2}. 
\end{align}
Applying Proposition \ref{prop_APS=sf}, we get the following. 
We introduce a real skew-adjoint operator $\hat{D}_m$ on $C^\infty (\mathbb{R}\times Y; S\oplus S)$
defined by
\begin{eqnarray}
\hat{D}_m :=\gamma  \partial_t +D_t=\left(
  \begin{array}{cc}
    \partial_t & D+m\hat{\kappa}^{\rm sm}_t  {\rm id}_S\\
    D-m\hat{\kappa}^{\rm sm}_t  {\rm id}_S & -\partial_t
\end{array}
  \right).
\end{eqnarray}  
Then we have
\begin{align}\label{eq_prf_main_ind=sf}
    \mathrm{Ind}(\hat{D}_m) = \mathrm{Sf}(\{D_t\}_{t \in [-1, 1]}) . 
\end{align}
By \eqref{eq_prf_main_sf} and \eqref{eq_prf_main_ind=sf}, we are left to prove the following. 
\begin{align}\label{eq_prf_main_goal}
    \mathrm{Ind}_{\rm APS}(D|_{Y_-}) = \mathrm{Ind}(\hat{D}_m). 
\end{align}
Now the proof is just a small modification of that of the main theorem of our previous work \cite{Fukaya:2019qlf}, so we summarize the main points here and refer the details to it. 
The strategy is to embed $Y_{\rm cyl}:= Y_- \cup [0,+\infty)\times X$ into $\R \times Y$ in a certain way, and use localization argument to prove \eqref{eq_prf_main_goal}.

First, since $D_X$ is assumed to be invertible, we can apply Proposition \ref{prop_APS=cyl} and get 
 \begin{align}\label{eq_APS=cyl}
      \mathrm{Ind}_{\rm APS}(D|_{Y_-})=\mathrm{Ind} (D_{\rm cyl}) . 
  \end{align}
Here  
$S$ and $D$ are extended to the cylinder to $S_{\rm cyl}$ and $D_{\rm cyl}$, respectively, in a canonical way. 

Moreover, let us consider a bundle $S_{\rm cyl}\oplus S_{\rm cyl}$ on
$\mathbb{R}_s\times Y_{\rm cyl}$
and introduce a higher-dimensional Dirac operator
\begin{eqnarray}
\hat{D}_{\rm cyl} = \left(
  \begin{array}{cc}
    \partial_s & D_{\rm cyl}+m\;{\rm sgn}\; {\rm id}_S,\\
    D_{\rm cyl}+m \;{\rm sgn}\; {\rm id}_S & -\partial_s
\end{array}
  \right),
\end{eqnarray}  
where ${\rm sgn}:\mathbb{R}\times Y_{\rm cyl}\to [-1,1]$ is the $L^\infty$-function\footnote{
Here we the operator is not smooth, but it causes no problem. 
We may also use a smoothing of the function sgn if we like. }
such that ${\rm sgn}=-1$ on $(-\infty,0)\times Y_{\rm cyl}$
and ${\rm sgn}=1$ on $(0,\infty)\times Y_{\rm cyl}$.
In the same way as \cite[Section 3.3]{Fukaya:2019qlf}
we have $\dim {\rm Ker} \hat{D}_{\rm cyl}=\dim {\rm Ker} D_{\rm cyl}$.

Now we recall the construction of a smooth embedding in \cite[Section 3.4]{Fukaya:2019qlf}. 
We define an embedding
\[
\bar{\tau}\colon (-2,2)\times Y_{\rm cyl} \to  \mathbb{R}\times Y. 
\]
Roughly speaking, the cylinder $\{0\} \times Y_{\rm cyl} \subset (-2, 2) \times Y_{\rm cyl}$ goes to a smoothing of the subset 
$\{0\} \times Y_{-} \cup [0, \infty) \times X \subset \R \times Y$. 

Let $R_1:=(-2,2)\times(-4,\infty)$ and $R_2=\mathbb{R}\times(-4,4)$.
We denote the coordinate of $R_1$ by $(-\tau,t)$, and that of $R_2$ by $(s,u)$.
Fix an embedding $\tau_{\mathbb{R}^2}:R_1\to R_2$ such that $\tau_{\mathbb{R}^2}\equiv {\rm id}$
for $t<-2$ and
\[
(-\tau,t)\mapsto (t,\tau)
\]
for $t>100$. Since $X$ has a collar neighbourhood isometric to $(-4,4)\times X$,
we can regard $R_1\times X$ and $R_2\times Y$ as open subsets of
$(-2,2)\times Y_{\rm cyl}$ and $\mathbb{R}\times Y$, respectively.
Using this, we can define an embedding $\bar{\tau}$
so that $\bar{\tau}\equiv {\rm id}_\mathbb{R}\times {\rm id}_Y$ on $(-2,2)\times Y_-$
and $\bar{\tau}\equiv \tau_{\mathbb{R}^2}\times {\rm id}_X$ on $R_1\times X$.
Note that $\bar{\tau}$ is an isometry outside a compact set $((-2,2)\times (-2,100)\times X)$.

Here, the important point for the localization argument is the following. 
We view $\bar{\tau}(\{0\}\times Y_{\rm cyl})$ as a domain-wall in $\R \times Y$, which separates $\R \times Y$ into two connected components. 
Then, the smooth function $\hat{\kappa}^{\rm sm} \colon \R \times Y \to [-1, 1]$ is a smoothing of the domain-wall function which takes value $\pm 1$ on the two connected components respectively. 

In our previous work, we have shown that there is a one-parameter family of
Riemannian metric connecting the induced metric by $\bar{\tau}$ and
the original one on $\mathbb{R}\times Y$ so that the low lying spectrum is unchanged
for $m>m_0$ with some real positive number $m_0$.
Thus for such $m$ we have
\begin{align}\label{eq_cyl=hatD}
    \dim {\rm Ker} \hat{D}_{\rm cyl}=\dim {\rm Ker} \hat{D}_m. 
\end{align}

By \eqref{eq_APS=cyl} and \eqref{eq_cyl=hatD}, we get \eqref{eq_prf_main_goal} and the result follows.


\section{Anomaly inflow and bulk-edge correspondence in the mod-two APS index}
\label{sec:inflow}

In the previous section, we have proved that for any mod-two
APS index of a Dirac operator on a manifold $Y_-$ with boundary $X$,
there exists a domain-wall Dirac fermion determinant
\begin{eqnarray}
\det(D_{\rm DW}D_{\rm PV}^{-1}) &=& \det \left(\frac{D+\kappa m {\rm id}_S}{D+m{\rm id}_S}\right),
\end{eqnarray}
and the quantity $(1-{\rm sgn}\det(D_{\rm DW}D_{\rm PV}^{-1})/2$ coincides with the original index (mod 2).
In the latter setup, instead of the boundary $X$,
we put the ``outside'' $Y_+$ to form a closed manifold $Y$,
and the mass term is introduced in such a way that
the sign flips at the original location of $X$.

Contrary to the original APS's massless Dirac operator,
which requires a non-local and unphysical boundary condition,
the operator $D$ in the domain-wall fermion determinant
is kept anti-Hermitian (skew-adjoint) without any difficulty.
The local and rotational symmetric boundary condition,
which is commonly expected in the fermion system of topological insulators,
is automatically satisfied on the domain-wall.
In this section, we discuss another physicist-friendly aspect
of the domain-wall fermion formulation: it allows a natural
decomposition of the index into bulk and edge contributions.

Let us introduce a free fermion field, or a trivial bundle $S_0$
on $Y$, where we assume by an appropriate regularization,
that $(S_0)_y$ the fiber at $y\in Y$, is isomorphic\footnote{
  For example, $(S_0)_y$ is isomorphic to $(S)_y$ at each site $y$ in
  the lattice regularization.} to $(S)_y$ (but we do not assume a smooth isomorphism on whole $Y$). 
Then we define the domain-wall Dirac operator with the opposite
sign of the mass to the original fermion: $\partial_{\rm DW}=\partial-\kappa m {\rm id}_{S_0}:C^\infty (Y;S_0)\to C^\infty (Y;S_0)$
with a free Dirac operator $\partial$.
As in the case with $D_{\rm DW}$, this new operator
$\partial_{\rm DW}$ also has edge-localized
eigenstates but with opposite chirality $\gamma_\tau=-1$.
Here we assume that $\partial_{\rm DW}$ is invertible,
which is achieved by, for instance, choosing a spin structure
such that the fermion obeys the anti-periodic boundary condition
around a nontrivial cycle on $X$.
We further assume that ${\rm sgn} \det \partial_{\rm DW}\partial_{\rm PV}^{-1}=+1$
with the free Pauli-Villars operator $\partial_{\rm PV}=\partial+ m {\rm id}_{S_0}$.
Namely, the corresponding mod-two index is always trivial.

Now we can decompose the ${\rm sgn}\det(D_{\rm DW}D_{\rm PV}^{-1})$ 
in Eq.~(\ref{eq:maintheorem})
as follows.
\begin{eqnarray}
  {\rm sgn}\det(D_{\rm DW}D_{\rm PV}^{-1}) &=& {\rm sgn}\left[\det(D_{\rm DW}D_{\rm PV}^{-1})\det(\partial_{\rm DW}\partial_{\rm PV}^{-1})\right]
  \nonumber\\
  &=& {\rm sgn}\left[\det\left(
    \begin{array}{cc}
      D_{\rm DW} & 0\\
      0 & \partial_{\rm DW}
      \end{array}\right)
    \det\left(
     \begin{array}{cc}
       D_{\rm PV}& 0 \\
       0 & \partial_{\rm PV}
     \end{array}
     \right)^{-1}\right]
  \nonumber\\
  &=& {\rm sgn}\left[\det D_{\rm edge}\right] {\rm sgn}\left[\det D_{\rm bulk}\right], 
\end{eqnarray}
where $D_{\rm edge/bulk}:C^\infty (Y;S\oplus S_0)\to C^\infty (Y;S\oplus S_0)$ are defined as 
\begin{eqnarray}
 D_{\rm edge}&:=&\left(
    \begin{array}{cc}
      D_{\rm DW} & 0\\
      0 & \partial_{\rm DW}
      \end{array}\right)\left(
    \begin{array}{cc}
      D_{\rm DW} & \mu I\\
      \mu I^{-1}& \partial_{\rm DW}
    \end{array}\right)^{-1},\\
 D_{\rm bulk}&:=&\left(
    \begin{array}{cc}
      D_{\rm DW} & \mu I\\
      \mu I^{-1}& \partial_{\rm DW}
    \end{array}
         \right)\left(
     \begin{array}{cc}
       D_{\rm PV}& 0 \\
       0 & \partial_{\rm PV}
     \end{array}
     \right)^{-1},   
\end{eqnarray}
with a positive constant $\mu$ and a trivial isomorphism
$I={\rm diag}(1,1...):(S_0)_y\to (S)_y$ at each $y\in Y$.
Note that both of $D_{\rm edge}$ and $D_{\rm bulk}$ are real operators,
and therefore, ${\rm sgn}\left[\det D_{\rm edge}\right]$ and
$ {\rm sgn}\left[\det D_{\rm bulk}\right]$
are both well-defined in the same sense as that for
the original operator $D_{\rm DW}D_{\rm PV}^{-1}$.

Now let us take a hierarchical limit $\lambda_{\rm edge}\ll \mu \ll m$,
where $\lambda_{\rm edge}$ denotes a typical energy scale of
the edge localized modes.
In this limit, $\det D_{\rm edge}$ is dominated by
contribution from the edge modes,
since $D_{\rm edge}$ operates as ${\rm id}_{S\oplus S_0}$
upto $\mu/m$ corrections on the bulk modes.
Similarly, $\det D_{\rm bulk}$ is essentially described by
the bulk modes.

It is important to remark here that
$D_{\rm edge/bulk}$ and their signs are not gauge invariant,
due to the new mass term $\mu I$ and its inverse.
Therefore, ${\rm sgn}\left[\det D_{\rm edge}\right]$
depends on the choice of the gauge, and
its gauge transformation can change its sign.
This is exactly what we expect for the global anomaly.
In their product ${\rm sgn}\left[\det D_{\rm edge}\right]{\rm sgn}\left[\det D_{\rm bulk}\right]$,
however, the $\mu$ dependence precisely cancels out and
the total index is gauge invariant.
Now we have manifestly achieved the global anomaly inflow,
decomposing the mod-two APS index:
\begin{eqnarray}
  {\rm Ind}_{\rm APS}(D|_{Y_-}) &=& I_{\rm edge} + I_{\rm bulk}\;\;\;(\mbox{mod 2}),\nonumber\\
  I_{\rm edge} &=& \frac{1-{\rm sgn}\left[\det D_{\rm edge}\right]}{2},\nonumber\\
  I_{\rm bulk} &=& \frac{1-{\rm sgn}\left[\det D_{\rm bulk}\right]}{2},
\end{eqnarray}
where the gauge dependence of $I_{\rm edge}$ is canceled by that of $I_{\rm bulk}$.
Or equivalently, we can say that the gauge invariance of the APS index guarantees the
bulk-edge correspondence of the global anomalies.

\section{Summary and discussion}
\label{sec:summary}

In this work, we gave a mathematical proof that for any APS index
${\rm Ind}_{\rm APS}(D)$ 
of a massless Dirac operator $D$ on a manifold $Y_-$
with boundary $X$, there exists a domain-wall Dirac fermion
determinant, whose sign coincides with $(-1)^{{\rm Ind}_{\rm APS}(D)}$.

Our domain-wall fermion Dirac operator is formulated
on a closed manifold extended from $Y_-$.
Instead, the mass term flips its sign at the original location of $X$.
Unlike the original APS setup, where an unphysical boundary condition
is needed to keep the chiral symmetry and edge localized modes
are not allowed to exist, the domain-wall fermion
keeps many essential features to understand the physics of topological matters.
No specific boundary condition is imposed {\it a priori},
but a local and physically sensible one having rotational symmetry
is automatically imposed on the domain-wall.
The distinction of the massless edge-localized modes and the bulk massive modes is manifest.
Moreover, we find a natural decomposition of the
mod-two APS index into edge and bulk contributions.
Each of them is given by a non-gauge invariant integer,
and therefore, contains a global anomaly.
The gauge invariance of the mod-two APS index guarantees
its cancellation or the bulk-edge correspondence of the global anomalies.
Thus, our theorem indicates that the domain-wall fermion determinant
(with Pauli-Villars regularization)
can be used as a physicist-friendly ``reformulation'' of the mod-two APS index.

The mathematical proof was given introducing a higher ($d+2$)-dimensional
Dirac operator $\hat{D}_m$\footnote{The physical role of $(d+2)$-dimensional Dirac operator was also discussed in our previous work \cite{Fukaya:2019qlf}.},
whose mod-two index is equal to the original ${\rm Ind}_{\rm APS}(D)$
and also equal to the spectral flow of a skew-adjoint operator $D_t$,
which coincides with $(1-{\rm sgn}\det(D_{\rm DW}D_{\rm PV}^{-1}))/2$.
What is the physical meaning of $\hat{D}_m$?
An interesting observation is that ${\rm Ind}(\hat{D}_m)$ is
equal to ${\rm Ind}_{\rm APS}(\hat{D}_m|_{Z_-})$, where $Z_-=Y\times [-1,1]$.
Then, denoting $Z=Y \times \mathbb{R}$ and $Z_+=Z\backslash Z_-$,
and introducing $\rho:Z\to [-1,1]$ by $\rho\equiv \pm 1$ on $Z_\pm\backslash Y$,
we can recursively use our main theorem
to obtain
\begin{eqnarray}
{\rm Ind}(\hat{D}_m) &=& \frac{1-{\rm sgn}\left[\det (\hat{D}_{\rm DW}\hat{D}_{\rm PV}^{-1})\right]}{2},
\end{eqnarray}  
where $\hat{D}_{\rm DW/PV}:C^\infty (Z;S\oplus S)\to C^\infty (Z;S\oplus S)$
are defined by $\hat{D}_{\rm DW}=\hat{D}_m+M\rho {\rm id}_{S\oplus S}$
and $\hat{D}_{\rm PV}=\hat{D}_m+M {\rm id}_{S\oplus S}$ respectively,
with a positive constant $M$, which is sufficiently larger than $m$.
This new domain-wall fermion Dirac operator
\begin{eqnarray}
 \hat{D}_{\rm DW}&:=&\left(
    \begin{array}{cc}
      \partial_t+M\rho {\rm id}_S& D+m\hat{\kappa}{\rm id}_S\\
     D-m\hat{\kappa}{\rm id}_S & -\partial_t+M\rho {\rm id}_S
    \end{array}\right)
\end{eqnarray}
in the large $m$ and $M$ limits,
has a ``edge-of-edge'' solution, whose asymptotic
behavior near $(\tau,t)=(0,1)$ is given by
\begin{eqnarray}
  \Psi(x,\tau,t) &=& \Phi(x)\exp(-m|\tau|)\exp(-M|t-1|),\nonumber\\
      {\rm id}\otimes\gamma_\tau \Phi(x)&=&\Phi(x),\;\;\; \gamma\otimes{\rm id}_{S}  \Phi(x)=-\Phi(x),\;\;\; D\Phi(x)=0.
\end{eqnarray}
Thus, our domain-wall fermion formulation naturally
contains a mathematical structure
that gapless states appear at a boundary of the system of codimension larger than one
\cite{Neuberger:2003yg,Fukaya:2016ofi,Hashimoto:2017tuh},
which may be useful to understand the physics of higher-order topological insulators \cite{hoti1,Schindler:2017etn}.

Another interesting application is
the formulation in lattice gauge theory\footnote{
For the standard APS index, index a lattice formulation was proposed in \cite{Fukaya:2019myi} using
the Wilson Dirac operator.
Here we consider the mod-two version.}.
On a flat Euclidean lattice with periodic boundary conditions,
of which continuum limit corresponds to $T^{d+1}$,
we can construct a lattice Dirac operator having the
same properties as $D_{\rm DW}$ above.
For example, in the $SU(2)$ gauge theory on a hyper-cubic 5-dimensional lattice $Y^{\rm lat}=L^5$,
the domain-wall Dirac operator $D_{\rm DW}^{\rm lat}:Y^{\rm lat}\otimes S^{\rm lat}\to Y^{\rm lat}\otimes S^{\rm lat}$
on a fermion field in the fundamental representation denoted by $S^{\rm lat}$
can be defined as
\begin{eqnarray}
D_{\rm DW}^{\rm lat}(x,y) &=& D_W(x,y) + \kappa m {\rm id}_{S^{\rm lat}}\delta_{x,y},
\end{eqnarray}  
where $x=(x_1,x_2,x_3,x_4,x_5)$ and $y=(y_1,y_2,y_3,y_4,y_5)$  represent discrete lattice points on $Y^{\rm lat}$,
$\kappa={\rm sgn}(x_5+1/2){\rm sgn}(L/2-x_5-1/2)$, the mass is in a range $0<m<2$ (to avoid contribution from doublers),
and $D_W(x,y)$ is the Wilson Dirac operator
\begin{eqnarray}
  D_W&=& \sum_{\mu=1}^5\gamma_\mu \frac{\nabla^f_\mu+\nabla^b_\mu}{2} -\sum_{\mu=1}^5\frac{\nabla^f_\mu\nabla^b_\mu}{2},\nonumber\\
  \nabla^f_\mu(x,y)&=& U_\mu(x)\delta_{x+1,y}-\delta_{x,y},\nonumber\\
  \nabla^b_\mu(x,y)&=& \delta_{x,y}-U_\mu^\dagger(y)\delta_{x-1,y}.
\end{eqnarray}
Here we take the lattice spacing unity.
Note that the link variables $U_\mu(x)$ in the fundamental representation of $SU(2)$
is pseudo-real: $U_\mu (x)^*=\mathcal{E} U_\mu (x)\mathcal{E}$
with the second Pauli matrix $\mathcal{E}=i\tau_2$, which is anti-symmetric.
This is also the case for $\gamma_\mu^* = C\gamma_\mu C$ with $C=\gamma_2\gamma_4\gamma_5$ (for the chiral representation),
which is also anti-symmetric.
Therefore, $D_{\rm DW}^{\rm lat}$ is real: $(D_{\rm DW}^{\rm lat})^*=C\mathcal{E}D_{\rm DW}^{\rm lat}C\mathcal{E}$.
Then we can ``define'' the mod-two APS index on the lattice by
\begin{eqnarray}
  \frac{1-{\rm sgn}\left[\det (D_{\rm DW}^{\rm lat})\right]}{2}\;\;\;\mbox{mod 2},
\end{eqnarray}
and it is natural to conjecture that this lattice index
for sufficiently large $L$ and smooth link variables
coincides with the continuum one on  $T^{4}\times [-L/2,0]$.
Note that the application to the mod-two AS index is straightforward, setting
$\kappa=-1$ to define the mod-two AS index on $T^5$ by
\begin{eqnarray}
  \frac{1-{\rm sgn}\left[\det (D_W - m {\rm id}_{S^{\rm lat}})\right]}{2}\;\;\;\mbox{mod 2}.
\end{eqnarray}

\begin{acknowledgements}
We thank N.~Kawai, Y.~Kikukawa, Y.~Kubota and K.~Yonekura for useful discussions.
This work was supported in part by JSPS KAKENHI (Grant numbers: JP15K05054, JP17H06461, JP17K14186, JP18H01216, JP18H04484,
JP18K03620, 19J20559 and 20K14307).
\end{acknowledgements}

\appendix
\section{An alternative proof of the main theorem}
In this appendix, we sketch an alternative simpler proof for the main theorem (Theorem \ref{theorem: main theorem}). 
The proof given here does not rely on the techniques developed in our previous work \cite{Fukaya:2019qlf}, and is self-contained. 
In particular, we do not use the embedding of cylinder $Y_{\rm cyl}:= Y_- \cup [0,+\infty)\times X$ into $\R \times Y$ or the localization argument.  

We use the same notations as in subsection \ref{subsec_proof}. 
We proceed in the same way to get \eqref{eq_prf_main_sf}, and using Proposition \ref{prop_APS=cyl}, we are left to prove the equality
\begin{align}
    \mathrm{Sf}(\{D_t\}_{t \in [-1, 1]}) = \mathrm{Ind}(D_{\rm cyl}). 
\end{align}
We proceed in a different way from here. 

First consider the following three operators acting on $L^2(Y_{\rm cyl}; S_{\rm cyl} \oplus S_{\rm cyl})$, 
\begin{eqnarray}
  D_{{\rm cyl}, -1} &:= \left(
  \begin{array}{cc}
    0 & D_{\rm cyl}+m  {\rm id}_S\\
    D_{\rm cyl}-m {\rm id}_S & 0
\end{array}
  \right), \\
  D_{{\rm cyl}, 0} &:= \left(
  \begin{array}{cc}
    0 & D_{\rm cyl}-m  {\rm id}_S\\
    D_{\rm cyl}+m {\rm id}_S & 0
\end{array}
  \right), \\
  D_{{\rm cyl}, 1} &:= \left(
  \begin{array}{cc}
    0 & D_{\rm cyl}+m \kappa^{\rm sm}_{\rm cyl} {\rm id}_S\\
    D_{\rm cyl}-m \kappa^{\rm sm}_{\rm cyl} {\rm id}_S & 0
\end{array}
  \right). 
\end{eqnarray}  
Here $\kappa_{\rm cyl}^{\rm sm} \colon Y_{\rm cyl} \to [-1, 1]$ is a smooth function with $\kappa_{\rm cyl}^{\rm sm} \equiv -1$ on $Y_-$ and $\kappa_{\rm cyl}^{\rm sm} \equiv 1$ on $X \times (4, +\infty)$. 
Let $\{D_{{\rm cyl}, t}^s\}_{(s, t) \in [0, 1] \times [-1, 1]}$ denote the two-parameter family of operators defined as 
\begin{align*}
    D_{{\rm cyl}, t}^s := \frac{1-t}{2}D_{{\rm cyl}, -1} + \frac{(1-s) (1+t)}{2}D_{{\rm cyl}, 0} + \frac{s(1+t)}{2}D_{{\rm cyl}, 1}. 
\end{align*}
This family consists of real and formally skew-adjoint operators. 
Moreover, by the invertibility of $D_X$ each operator is Fredholm, and the family is continuous, by the same argument as that in \cite[Section 2]{Bunke1994}. 
We regard this as a path, parameterized by $s \in [0, 1]$, of paths $\{D_{{\rm cyl}, t}^s\}_{t \in [-1, 1]}$ of real skew-adjoint Fredholm operators. 
Obviously, $D_{{\rm cyl}, -1}^s = D_{{\rm cyl}, -1} $ is invertible. 
Using the invertibility of $D_X$, for $m$ large enough, we can also see that $D^s_{{\rm cyl}, 1}$ are all invertible for all $s \in [0, 1]$: 
this can be shown in the same way as \cite[Proposition 9]{Fukaya:2019qlf}. 
Thus, by the deformation invariance of spectral flows, we get
\begin{align}
    \mathrm{Sf}(\{D^0_{{\rm cyl}, t}\}_{t \in [-1, 1]}) = \mathrm{Sf}(\{D^1_{{\rm cyl}, t}\}_{t \in [-1, 1]}). 
\end{align}
Moreover, at $s = 0$, we see directly from the definition of spectral flow that
\begin{align}
    \mathrm{Sf}(\{D^0_{{\rm cyl}, t}\}_{t \in [-1, 1]}) =\mathrm{Ind}(D_{\rm cyl}) . 
\end{align}
So we get
\begin{align}
    \mathrm{Sf}(\{D^1_{{\rm cyl}, t}\}_{t \in [-1, 1]}) =\mathrm{Ind}(D_{\rm cyl}) . 
\end{align}
Note that, restricted on the cylindrical end $X \times (4, \infty)$, the family $\{D^1_{{\rm cyl}, t}\}_t$ does not depend on $t$. 

In order to pass to the closed manifold $Y$, we consider the manifold $Y_{+, \rm cyl} :=(-\infty, 0)\times X \cup Y_+$ with the corresponding bundle $S_{+, {\rm cyl}}$ and $D_{+ , {\rm cyl}}$. 
Let $\{D_{+, {\rm cyl}, t}\}_{t \in [-1, 1]}$ the constant family of operators on $L^2(Y_{+, \rm cyl} ; S_{+, \rm cyl} \oplus S_{+, \rm cyl})$ defined by
\begin{eqnarray}
  D_{+, {\rm cyl}, t} &:= \left(
  \begin{array}{cc}
    0 & D_{+, \rm cyl}+m  {\rm id}_S\\
    D_{+, \rm cyl}-m {\rm id}_S & 0
\end{array}
  \right). 
\end{eqnarray}
Of course we have
\begin{align}
    \mathrm{Sf}(\{D_{+, {\rm cyl}, t}\}_{t \in [-1, 1]}) =0. 
\end{align}

By the gluing property of Fredholm index, we can show the corresponding gluing formula for mod-two spectral flows\footnote{
One simplest way to show the gluing of spectral flows here is to use Proposition \ref{prop_APS=sf}. 
Using it, we can reduce the problem to the gluing property of indices of operators on $Y_- \times \R$ and $Y_+ \times \R$, which is standard. }. 
If we glue the family $\{D^1_{{\rm cyl}, t}\}_{t \in [-1, 1]}$ and $\{D_{+, {\rm cyl}, t}\}_{t \in [-1, 1]}$ along $X$, 
we get the family $\{D_t\}_{t \in [-1, 1]}$ on $L^2(Y; S \oplus S)$ defined in \eqref{eq_Dt}, and get
\begin{align}
\mathrm{Sf}(\{D_t\}_{t \in [-1, 1]})
= \mathrm{Sf}(\{D^1_{{\rm cyl}, t}\}_{t \in [-1, 1]}) + \mathrm{Sf}(\{D_{+, {\rm cyl}, t}\}_{t \in [-1, 1]}) 
= \mathrm{Ind}(D_{\rm cyl}) . 
\end{align}
So the proof is complete. 

\begin{remark}
The authors came up with this simpler proof while writing this paper. 
We can also prove the main theorem of our previous work \cite{Fukaya:2019qlf} in a similar way. 

Actually, the two proofs are essentially the same. 
The relation between them can be understood by comparing Figures \ref{fig_1}, \ref{fig_2} and \ref{fig_3}. 
Figure \ref{fig_1} corresponds to the proof in the Appendix, and Figure \ref{fig_3} corresponds to that in subsection \ref{subsec_proof}. 
On the pink regions we have $\kappa=1$, and on the white regions we have $\kappa = -1$. 
In the proofs, we identified the mod-two APS index with the mod-two spectral flows between operators defined on red and blue parts. 
The fact that spectral flows in Figures \ref{fig_1} and \ref{fig_3} coincide can be understood by moving the red and blue manifolds in the way shown in Figure \ref{fig_2}. 
\end{remark}

\begin{figure}[htbp]
\begin{center}
  \begin{tikzpicture}
    \pic {infiniteCylinderX};
    \filldraw [red, nearly transparent] (6,1) -- (0,1) -- (0,0) -- (6,0);
    \filldraw [gray] (6,0.125) -- (0.125,0.125) -- (0.125,0.5335) arc [start angle=-60, end angle=240, radius=0.25] -- (-0.125,-0.125) -- (6,-0.125);

    \draw [blue] [ultra thick, rounded corners] (1,1) -- (1,0.5) -- (6,0.5);
    \draw [red] [ultra thick, rounded corners] (-1,1) -- (-1,-0.5) -- (6,-0.5);
  \end{tikzpicture}
  \caption{The proof in the Appendix}
  \label{fig_1}
\end{center}

\begin{center}
  \begin{tikzpicture}
    \pic {infiniteCylinderX};
    \filldraw [red, nearly transparent] (6,1) -- (0,1) -- (0,0) -- (6,0);
    \filldraw [gray] (6,0.125) -- (0.125,0.125) -- (0.125,0.5335) arc [start angle=-60, end angle=240, radius=0.25] -- (-0.125,-0.125) -- (6,-0.125);

    \draw [blue] [ultra thick, rounded corners] (1,1) -- (1,-0.4) -- (6,-0.4);
    \draw [red] [ultra thick, rounded corners] (-1,1) -- (-1,-0.5) -- (6,-0.5);
  \end{tikzpicture}
  \caption{}
  \label{fig_2}
\end{center}

\begin{center}
  \begin{tikzpicture}
    \pic {infiniteCylinderX};
    \filldraw [red, nearly transparent] (6,1) -- (0,1) -- (0,0) -- (6,0);
    \filldraw [gray] (6,0.125) -- (0.125,0.125) -- (0.125,0.5335) arc [start angle=-60, end angle=240, radius=0.25] -- (-0.125,-0.125) -- (6,-0.125);

    \draw [blue] [ultra thick, rounded corners] (1,1) -- (1,-1);
    \draw [red] [ultra thick, rounded corners] (-1,1) -- (-1,-1);
  \end{tikzpicture}
  \caption{The proof in subsection \ref{subsec_proof}}
  \label{fig_3}
\end{center}

\end{figure}

\end{document}